\documentclass[a4paper,UKenglish,cleveref, autoref, thm-restate]{lipics-v2021}

\pdfoutput=1 
\hideLIPIcs  

\title{Galactic Token Sliding\footnote{This work is supported by PHC Cedre project 2022 "PLR".}} 

\author{Valentin Bartier}
{Univ Lyon, CNRS, ENS de Lyon, Université Claude Bernard Lyon 1, LIP UMR5668, France}
{valentin.bartier@grenoble-inp.fr}{}{Supported by ANR project GrR (ANR-18-CE40-0032).}

\author{Nicolas Bousquet}
{Univ Lyon, CNRS, UCBL, INSA Lyon, LIRIS, UMR5205, France}
{nicolas.bousquet@univ-lyon1.fr}{https://orcid.org/0000-0003-0170-0503}{Supported by ANR project GrR (ANR-18-CE40-0032).}

\author{Amer~E.~Mouawad}
{Department of Computer Science, American University of Beirut, Lebanon \and 
Department of Informatics, University of Bremen, Germany}
{aa368@aub.edu.lb}{https://orcid.org/0000-0003-2481-4968}{Research supported by
the Alexander von Humboldt Foundation and partially supported by URB project ``A theory of change through the lens of reconfiguration''.}

\authorrunning{V.~Bartier, N.~Bousquet, A.E.~Mouawad} 

\Copyright{Valentin Bartier, Nicolas Bousquet, Amer~E.~Mouawad} 

\ccsdesc[500]{Theory of computation~Fixed parameter tractability}
\ccsdesc[500]{Theory of computation~W hierarchy}

\keywords{reconfiguration, independent set, galactic reconfiguration, sparse graphs, token sliding, parameterized complexity} 

\category{} 

\relatedversion{} 

\EventEditors{John Q. Open and Joan R. Access}
\EventNoEds{2}
\EventLongTitle{42nd Conference on Very Important Topics (CVIT 2016)}
\EventShortTitle{CVIT 2016}
\EventAcronym{CVIT}
\EventYear{2016}
\EventDate{December 24--27, 2016}
\EventLocation{Little Whinging, United Kingdom}
\EventLogo{}
\SeriesVolume{42}
\ArticleNo{23}

\usepackage{todonotes}
\newtheorem{question}{Question}

\numberwithin{theorem}{section}
\numberwithin{observation}{section}
\numberwithin{proposition}{section}
\numberwithin{lemma}{section}
\numberwithin{corollary}{section}
\numberwithin{definition}{section}
\numberwithin{claim}{section}
\numberwithin{remark}{section}
\numberwithin{question}{section}

\newcommand{\WO}{\textsf{W[1]}}
\newcommand{\PP}{\textsf{P}}
\newcommand{\PPSPACE}{\textsf{PSPACE}}
\newcommand{\NPP}{\textsf{NP}}
\newcommand{\FPT}{\textsf{FPT}}
\newcommand{\Oh}{\mathcal{O}}

\begin{document}

\nolinenumbers
\maketitle
\begin{abstract}
Given a graph $G$ and two independent sets $I_s$ and $I_t$ of size $k$, the \textsc{Independent Set Reconfiguration} problem asks whether there exists a sequence of independent sets (each of size $k$) $I_s = I_0, I_1, I_2, \ldots, I_\ell = I_t$ such that each independent set is obtained from the previous one using a so-called reconfiguration step. Viewing each independent set as a collection of $k$ tokens placed on the vertices of a graph $G$, the two most studied reconfiguration steps are token jumping and token sliding. In the \textsc{Token Jumping} variant of the problem, a single step allows a token to jump from one vertex to any other vertex in the graph. In the \textsc{Token Sliding} variant, a token is only allowed to slide from a vertex to one of its neighbors. Like the \textsc{Independent Set} problem, both of the aforementioned problems are known to be W[1]-hard on general graphs. A very fruitful line of research~\cite{goos_dynamic_1988, grohe_deciding_2017, telle_fpt_2019, pilipczuk_kernelization_2021} has showed that the \textsc{Independent Set} problem becomes fixed-parameter tractable when restricted to sparse graph classes, such as  planar, bounded treewidth, nowhere-dense, and all the way to biclique-free graphs. Over a series of papers, the same was shown to hold for the \textsc{Token Jumping} problem~\cite{ahn_fixed-parameter_2014, lokshtanov_reconfiguration_2018, siebertz_reconfiguration_2018, klasing_token_2017}. As for the \textsc{Token Sliding} problem, which is mentioned in most of these papers, almost nothing is known beyond the fact that the problem is polynomial-time solvable on trees \cite{demaine_linear-time_2015} and interval graphs~\cite{DBLP:conf/wg/BonamyB17}. 
We remedy this situation by introducing a new model for the reconfiguration of independent sets, which we call galactic reconfiguration. Using this new model, we show that (standard) \textsc{Token Sliding} is fixed-parameter tractable on graphs of bounded degree, planar graphs, and chordal graphs of bounded clique number. We believe that the galactic reconfiguration model is of independent interest and could potentially help in resolving the remaining open questions concerning the (parameterized)  complexity of \textsc{Token Sliding}. 
\end{abstract}

\section{Introduction}\label{sec-intro}

Many algorithmic questions can be represented as follows:
given the description of a system state and the description of a state we would ``prefer'' the system to be in, is it possible to transform the system
from its current state into the more desired one without ``breaking'' the system in the process? And if yes, how many steps are needed? Such problems naturally arise in the fields of mathematical puzzles, operational research, computational geometry~\cite{DBLP:journals/comgeo/LubiwP15}, bioinformatics,  and quantum computing~\cite{DBLP:conf/icalp/GharibianS15} for instance. These questions received a substantial amount of attention under the so-called \emph{combinatorial reconfiguration framework} in the last few years~\cite{DBLP:journals/tcs/BrewsterMMN16,H13,Wrochna15}.
We refer the reader to the surveys by van den Heuvel~\cite{H13}
and Nishimura~\cite{DBLP:journals/algorithms/Nishimura18} for more background on combinatorial reconfiguration.

\subparagraph*{Independent set reconfiguration.} In this work, we focus on the reconfiguration of independent sets.
Given a simple undirected graph $G$, a set of vertices $S \subseteq V(G)$ is an \emph{independent set} if the vertices
of $S$ are all pairwise non-adjacent. Finding an independent set of maximum cardinality, i.e., the {\sc Independent Set} problem,
is a fundamental problem in algorithmic graph theory and is known to be not only \NPP-hard, but also \WO-hard and
not approximable within $\Oh(n^{1-\epsilon})$, for any $\epsilon > 0$, unless $\PP = \NPP$~\cite{DBLP:journals/toc/Zuckerman07}.
Moreover, {\sc Independent Set} is known to remain \WO-hard on graphs excluding $C_4$ (the cycle on four vertices) as an induced subgraph~\cite{DBLP:conf/iwpec/BonnetBCTW18}.

We view an independent set as a collection of tokens placed on the vertices of a graph such that no two tokens are adjacent.
This gives rise to two natural adjacency relations between independent sets (or token configurations), also called \emph{reconfiguration steps}.
These two reconfiguration steps, in turn, give rise to two combinatorial reconfiguration problems.
In the {\sc Token Jumping} (TJ) problem, introduced by Kami\'{n}ski et al.~\cite{KMM12}, a single reconfiguration step consists of first removing a token on some vertex $u$ and then immediately adding it back on any other vertex $v$, as long as no two tokens become adjacent. The token is said to \emph{jump} from vertex $u$ to vertex $v$.
In the {\sc Token Sliding} (TS) problem, introduced by Hearn and Demaine~\cite{DBLP:journals/tcs/HearnD05}, two independent sets are adjacent if one can be obtained from the other by a token jump from vertex $u$ to vertex $v$ with the additional requirement of $uv$ being an edge of the graph. The token is then said to \emph{slide} from vertex $u$ to vertex $v$ along the edge $uv$. Note that, in both the TJ and TS problems, the size of independent sets is fixed.
Generally speaking, in the {\sc Token Jumping} and {\sc Token Sliding} problems,
we are given a graph $G$ and two independent sets $I_s$ and $I_t$ of $G$. The goal is to determine
whether there exists a sequence of reconfiguration steps -- a \emph{reconfiguration sequence} -- that
transforms $I_s$ into $I_t$ (where the reconfiguration step depends on the problem).

Both problems have been extensively studied, albeit under
different names~\cite{DBLP:conf/wg/BonamyB17,DBLP:conf/swat/BonsmaKW14,demaine_linear-time_2015,DBLP:conf/isaac/Fox-EpsteinHOU15,DBLP:conf/tamc/ItoKOSUY14,DBLP:journals/ieicet/ItoNZ16,KMM12,lokshtanov_reconfiguration_2018,lokshtanov_reconfiguration_2018}.
It is known that both problems are \PPSPACE-complete, even on restricted graph classes such as
graphs of bounded bandwidth (and hence pathwidth)~\cite{WROCHNA14} and planar graphs~\cite{DBLP:journals/tcs/HearnD05}. 

On the positive side, it is easy to prove that {\sc Token Jumping} can be decided in polynomial time on trees (and even on chordal graphs) since we simply have to iteratively move tokens on leaves (resp. vertices that only appear in the bag of a leaf in the clique tree) to transform an independent set into another. Unfortunately, for {\sc Token Sliding}, the problem becomes more complicated because of what we call the \emph{bottleneck effect}. Indeed, there might be a lot of empty leaves in the tree but there might be a bottleneck in the graph that prevents us from reaching these desirable vertices. For instance, if we imagine a star plus a long subdivided path attached to the center of the star. One cannot move any token from leaves of the star to the path if there are at least two leaves in the independent set. 
Even if we can overcome this issue for instance on trees~\cite{demaine_linear-time_2015} and on interval graphs~\cite{DBLP:conf/wg/BonamyB17}, the {\sc Token Sliding} problem remains much ``harder'' than the {\sc Token Jumping} problem. In split graphs for instance (which are chordal), {\sc Token Sliding} is PSPACE-complete~\cite{DBLP:conf/stacs/Belmonte0LMOS19}. Lokshtanov and Mouawad~\cite{DBLP:journals/talg/LokshtanovM19} showed that, in bipartite graphs, {\sc Token Jumping} is \NPP-complete
while {\sc Token Sliding} remains \PPSPACE-complete. 

In this paper we focus on the parameterized complexity of the {\sc Token Sliding} problem. While the complexity of {\sc Token Jumping} parameterized by the size of the independent set is quite well understood, the comprehension of the complexity of {\sc Token Sliding} remains evasive.

A problem $\Pi$ is \FPT\ (Fixed Parameterized Tractable) parameterized by $k$ if one can solve it in time $f(k) \cdot poly(n)$, for some computable function $f$. In other words, the combinatorial
explosion can be restricted to a parameter $k$. In the rest of the paper, our parameter $k$ will be the size of the independent set (i.e.\ number of tokens).
Both {\sc Token Jumping} and {\sc Token Sliding} are known to be \WO-hard\footnote{Informally, it means that they are very unlikely to admit an \FPT\ algorithm.}
parameterized by $k$ on general graphs~\cite{lokshtanov_reconfiguration_2018}.

On the positive side, Lokshtanov et al.~\cite{lokshtanov_reconfiguration_2018} showed that {\sc Token Jumping} is \FPT\ on bounded degree graphs. This result has been extended in a series of papers to planar graphs, nowhere-dense graphs, and finally strongly $K_{\ell,\ell}$-free graphs~\cite{IKO14,klasing_token_2017}, a graph
being strongly $K_{\ell,\ell}$-free if it does not contain any $K_{\ell,\ell}$ as a subgraph. 

For {\sc Token Sliding}, it was proven in~\cite{DBLP:journals/algorithmica/BartierBDLM21} that the problem is \WO-hard on bipartite graphs and $C_4$-free graphs (a similar result holds for {\sc Token Jumping} but based on  weaker assumptions for the bipartite case~\cite{DBLP:conf/iwpec/AgrawalAD21}).

However, almost no positive result is known for {\sc Token Sliding} even for incredibly simple cases like bounded degree graphs. Our main contributions are to develop two general tools for the design of parameterized algorithms for {\sc Token Sliding}, namely galactic reconfiguration and types. Galactic reconfiguration is a general simple tool that allows us to reduce instances. Using it, we will derive that {\sc Token Sliding} is \FPT\ on bounded degree graphs.
Our second tool, called types, will in particular permit to show that the deletion of a small subset of vertices leaves too many components, then one of them can be removed. 
Combining both tools with additional rules, we prove that {\sc Token Sliding} is \FPT\ on planar graphs and on chordal graphs of bounded clique number. We complement these results by proving that {\sc Token Sliding} is \WO-hard on split graphs.

\subparagraph{Galactic reconfiguration.} 
Our first result is the following:

\begin{theorem}\label{thm:ts-bounded-degree-fpt-intro}
{\sc Token Sliding} is \FPT\ on bounded degree graphs parameterized by $k$.
\end{theorem}

Much more than the result itself, we believe that our main contribution here is the general framework we developed for its proof, namely galactic reconfiguration.
Before explaining exactly what it is, let us explain the intuition behind it. As we already said, even if there are independent vertices which are far apart from the vertices of an independent set, we are not sure we can reach them because of the bottleneck effect. Nevertheless, it \emph{should} be possible to reduce a part the graph that does not contain any token as we can find irrelevant vertices when we have a large grid minor since, when we enter in the structure, we can basically move as we want in it (and then avoid to put tokens close to each other). 
However, proving that a structure can be  reduced in reconfiguration is usually very technical. To overcome this problem, we actually introduce a new type of vertices called \emph{black holes} which can swallow as many tokens of the independent set as we like. 
A \emph{galactic graph} is a graph that might contain black holes. A \emph{galactic independent set} is a set of vertices on which tokens lie, such that the set of non black-hole vertices holding tokens is an independent set and such that each black-hole might contain any number of tokens.

Our main result is to prove that if there exists a long shortest path that is at distance two from the initial and target independent sets, then we can replace it by a black hole (whose neighborhood is the union of the neighborhoods of the path).  This rule, together with other simple rules on galactic graphs, allows us to reduce the size of bounded-degree graphs until they reach a size of at most $f(k)$, for some computable function $f$, in polynomial time. Since a kernel ensures the existence of an \FPT\ algorithm, Theorem~\ref{thm:ts-bounded-degree-fpt-intro} holds.

\subparagraph*{Types and the multi-component reduction.}
In the rest of the paper, we combine galactic graphs with other techniques to prove that {\sc Token Sliding} is \FPT\ on several other graph classes.  We first prove the following:

\begin{theorem}\label{thm:ts-planar-intro}
{\sc Token Sliding} is \FPT\ on  planar graphs parameterized by $k$.
\end{theorem}

To prove Theorem~\ref{thm:ts-planar-intro}, we cannot simply use our previous long path construction since, in a planar graph, there might be a universal vertex which prevents the existence of a long shortest path. Note that the complexity of {\sc Token Sliding} is open on outerplanar graphs and it was not known to be \FPT\ prior to our work.

Our strategy consists in reducing to planar graphs of bounded degree and then applying Theorem~\ref{thm:ts-bounded-degree-fpt-intro}.
To do so, we provide some general tools to reduce graphs for {\sc Token Sliding}. Namely, we show that if we have a set $X$ of vertices such that $G - X$ contains too many connected components (in terms of $k$ and $|X|$) then at least one of them can be safely removed. 

The idea of the proof consists in defining the \emph{type} of a connected component of $G-X$. From a very high level perspective, the type of a path in a component of $G - X$ as the sequence of its neighborhoods in $X$ \footnote{The exact definition is actually more complicated.}. The type of a component $C$ is the union of the types of the paths starting on a vertex of $C$. We then show that if too many components of $G-X$ have the same type then one of them can be removed.

However, this component reduction is not enough since, in the case of a vertex universal to an outerplanar graph we do not have a lot of components when we remove the universal vertex. We prove that, we can also reduce a planar graph if (i) there are too many vertex-disjoint $(x,y)$-paths for some pair $x,y$ of vertices or (ii) if a vertex has too many neighbors on an induced path. Since one can prove that in an arbitrarily large planar graph with no long shortest path (i) or (ii) holds, it will imply Theorem~\ref{thm:ts-planar-intro}.
Note that our proof techniques can be easily adapted to prove that the problem is \FPT\ for any graph of bounded genus. 
We think that the notion of types may be crucial to derive \FPT\ algorithms on larger classes of graphs such as bounded treewidth graphs.

We finally provide another application of our method by proving that the following holds:

\begin{theorem}\label{thm:ts-chordal-bounded-clique-intro}
{\sc Token Sliding} is \FPT\ on chordal graphs of bounded clique number.
\end{theorem}

The proof of Theorem~\ref{thm:ts-chordal-bounded-clique-intro} consists in proving that, since there is a long path in the clique tree, we can either find a long shortest path (and we can reduce the graph using galactic graphs) or find a vertex $x$ in a large fraction of the bags of this path. In the second case, we show that we can again reduce the graph.
%
%
We complement this result by proving that it cannot be extended to split graphs, contrarily to {\sc Token Jumping}.

\begin{theorem}
{\sc Token Sliding} is \WO-hard on split graphs.
\end{theorem}

We show hardness via a reduction from the \textsc{Multicolored Independent Set} problem, known to be \WO-hard~\cite{DBLP:books/sp/CyganFKLMPPS15}. The crux of the reduction relies on the fact that we have a clique of unbounded size and hence we can use different subsets of the clique to encode vertex selection gadgets and non-edge selection gadgets. A summary of the current parameterized complexity status of {\sc Token Jumping} and {\sc Token Sliding} is depicted in Figure~\ref{fig:classes}.

\begin{figure}
    \centering
    \includegraphics[scale=0.7]{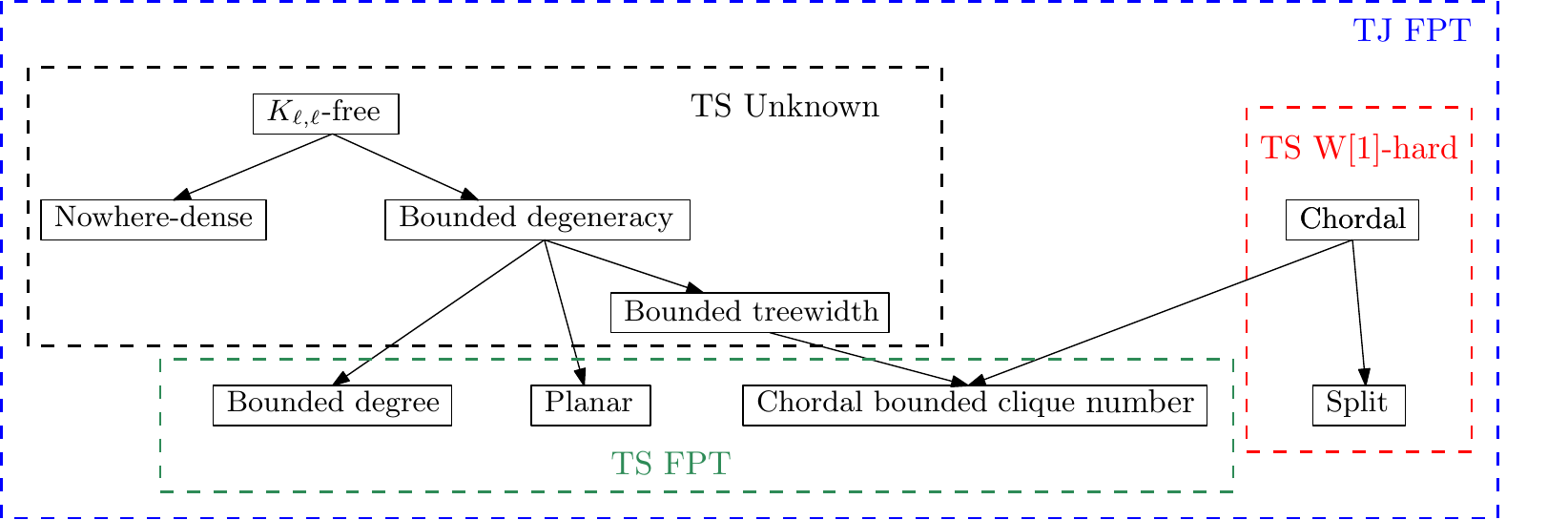}
    \caption{An overview of the parameterized complexity status of {\sc Token Jumping} and {\sc Token Sliding} on different graph classes.}
    \label{fig:classes}
\end{figure}

\subparagraph*{Further work.}
The first natural generalization of our result on chordal graphs of bounded clique size would be the following:

\begin{question}\label{q:treewidth}
Is {\sc Token Sliding} \FPT\ on bounded treewidth graphs? Or simpler, how about on bounded pathwidth graphs?
\end{question}
Recall that the problem is \PPSPACE-complete on graphs of constant bandwidth for a large enough constant that is not explicit in the proof~\cite{WROCHNA14}. 
Note that our galactic reconfiguration rules directly ensure that {\sc Token Sliding} is \FPT\ for graphs of bounded bandwidth. Our multi-component reduction ensures that the problem is \FPT\ for graphs of bounded treedepth. But for bounded pathwidth, the situation is unclear. There are good indications to think that solving the bounded pathwidth case is actually the hardest step to obtain an \FPT\ algorithm for graphs of bounded treewidth.
On the positive side, we simply know that the problem is polynomial time solvable on graphs of treewidth one (namely forests)~\cite{demaine_linear-time_2015} and the problem is open for graphs of treedwidth $2$. It is even open for outerplanar graphs:

\begin{question}
Is {\sc Token Sliding} polynomial-time solvable on outerplanar graphs? How about triangulated outerplanar graphs?
\end{question}

We did not succeed in answering Question~\ref{q:treewidth} but we think that the method we used for Theorem~\ref{thm:ts-chordal-bounded-clique-intro} is a good starting point but the analysis is much more involved.
In the case of {\sc Token Jumping} the problem is actually \FPT\ on strongly $K_{\ell,\ell}$-free graphs which contains planar graphs, bounded treewidth graphs, and many other classes. If the answer to Question~\ref{q:treewidth} is positive, the next step to achieve a similar statement for {\sc Token Sliding} would consist in looking at minor-free graphs and, more generally, nowhere dense graphs.

\subparagraph*{Organization of the paper.} In Section~\ref{sec-galactic}, we formally introduce galactic graphs and provide our main reduction rules concerning such graphs including the long short path reduction lemma. In Section~\ref{sec-components}, we introduce the notion of types and journeys and prove that if there are too many connected components in $G - X$ then at least one of them can be removed. In Section~\ref{sec-planar}, we prove that {\sc Token Sliding} is \FPT\ on planar graphs. We prove the same for chordal graphs of bounded clique number in Section~\ref{sec-chordal} and finally give our hardness reduction for split graphs in Section~\ref{sec-split}.

\section{Preliminaries}\label{sec-prelim}
We denote the set of natural numbers by $\mathbb{N}$.
For $n \in \mathbb{N}$ we let $[n] = \{1, 2, \dots, n\}$.

We assume that each graph $G$ is finite, simple, and undirected. We let $V(G)$ and $E(G)$ denote the vertex set and edge set of $G$, respectively.
The {\em open neighborhood} of a vertex $v$ is denoted by $N_G(v) = \{u \mid uv \in E(G)\}$ and the
{\em closed neighborhood} by $N_G[v] = N_G(v) \cup \{v\}$.
For a set of vertices $Q \subseteq V(G)$, we define $N_G(Q) = \{v \not\in Q \mid uv \in E(G), u \in Q\}$ and $N_G[Q] = N_G(Q) \cup Q$.
The subgraph of $G$ induced by $Q$ is denoted by $G[Q]$, where $G[Q]$ has vertex set
$Q$ and edge set $\{uv \in E(G) \mid u,v \in Q\}$. We let $G - Q = G[V(G) \setminus Q]$.

A {\em walk} of length $\ell$ from $v_0$ to $v_\ell$ in $G$ is a vertex sequence $v_0, \ldots, v_\ell$, such that
for all $i \in \{0, \ldots, \ell-1\}$, $v_iv_{i + 1} \in E(G)$.
It is a {\em path} if all vertices are distinct. It is a {\em cycle}
if $\ell \geq 3$, $v_0 = v_\ell$, and $v_0, \ldots, v_{\ell - 1}$ is a path.
A path from vertex $u$ to vertex $v$ is also called a {\em $uv$-path}.
For a pair of vertices $u$ and $v$ in $V(G)$, by $\textsf{dist}_G(u,v)$ we denote the {\em distance} or length of a shortest $uv$-path
in $G$ (measured in number of edges and set to $\infty$ if $u$ and $v$ belong to different connected components).
The {\em eccentricity} of a vertex $v \in V(G)$, $\textsf{ecc}(v)$, is equal to $\max_{u \in V(G)}(\textsf{dist}_G(u,v))$.
The {\em diameter} of $G$, $\textsf{diam}(G)$, is equal to $\max_{v \in V(G)}(\textsf{ecc}(v))$.

\section{Galactic graphs and galactic token sliding}\label{sec-galactic}

We say that a graph $G = (V, E)$ is a \emph{galactic graph} when $V(G)$ is partitioned into two sets $A(G)$ and $B(G)$ where the set $A(G) \subseteq V(G)$ is the set of vertices that we call \emph{planets} and the set $B(G) \subseteq V(G)$ is the set of vertices that we call \emph{black holes}.
For a given graph $G'$, we write $G' \prec G$ whenever $|A(G')| < |A(G)|$ or, in case of equality, $|B(G')| \leq |B(G)|$.
In the standard {\sc Token Sliding} problem, tokens are restricted to sliding along edges of a graph as long as the resulting sets remain independent.
This implies that no vertex can hold more than one token and no two tokens can ever become adjacent.
In a galactic graph, the rules of the game are slightly modified. When a token reaches a black hole (a special kind of vertex), the token is
\emph{absorbed} by the black hole. This implies that a black hole can hold more than one token, in fact it can hold all $k$ tokens.
Moreover, we allow tokens to be adjacent as long as one of the two vertices is a black hole (since black holes are assumed to make tokens ``disappear''). 
On the other hand, a black hole can also ``project'' any of the tokens it previously absorbed onto any vertex in its neighborhood (be it a planet or a black hole).
Of course, all such moves require that we remain an independent set in the galactic sense.  
We say that a set $I$ is a \emph{galactic independent set} of a galactic graph $G$ whenever $G[I \cap A]$ is independent. To fully specify a galactic independent
set $I$ of size $k$ containing more than one token on black holes, we use a weight function $\omega_I: V(G) \rightarrow \{0, \ldots, k\}$. 
Hence, $\omega_I(v) \leq 1$ whenever $v \in A(G)$, $\omega_I(v) \in \{0, \ldots, k\}$ whenever $v \in B(G)$,
and $\sum_{v \in V(G)}{\omega_I(v)} = k$.

We are now ready to define the {\sc Galactic Token Sliding} problem. We are given a galactic graph $G$, an integer $k$, and two galactic independent sets $I_s$ and $I_t$
such that $|I_s| = |I_t| = k \geq 2$ (when $k = 1$ the problem is trivial). 
The goal is to determine whether there exists a sequence of token slides
that will transform $I_s$ into $I_t$ such that each intermediate set remains a galactic independent set. As for the classical \textsc{Token Sliding} problem, given a galactic graph $G$ we can define a reconfiguration graph which we call the \emph{galactic reconfiguration graph} of $G$. It is the graph whose vertex set is the set of all galactic independent sets of $G$, two vertices being adjacent if their corresponding galactic independent sets differ by exactly one token slide.
We always assume the input graph $G$ to be a connected graph, since
we can deal with each component independently otherwise.
Furthermore, components without tokens can be safely deleted.
Given an instance $(G,k,I_s,I_t)$ of {\sc Galactic Token Sliding}, we say that $(G,k,I_s,I_t)$ can be \emph{reduced} if we
can find an instance $(G',k',I_s',I_t')$ which is positive (a yes-instance) if and only if $(G,k,I_s,I_t)$ is positive (a yes-instance) and $G' \prec G$.

Let $G$ be a galactic graph. A \emph{planetary component} is a maximal connected component of $G[A]$. A \emph{planetary path} $P$, or \emph{$A$-path},
composed only of vertices of $A$, is called \emph{$A$-geodesic} if, for every $x,y$ in $P$, $\textsf{dist}_{G[A]}(x,y) = \textsf{dist}_P(x,y)$.
We use the term \emph{$A$-distance} to denote the length of a shortest path between vertices $u,v \in A$ such that all vertices of the path are also in $A$.
Let us state a few reduction rules that allow us to safely reduce an instance $(G,k,I_s,I_t)$ of {\sc Galactic Token Sliding} to an instance $(G',k',I_s',I_t')$. 

\begin{itemize}
\item Reduction rule R1 (adjacent black holes rule): If two black holes $u$ and $v$ are adjacent, we contract them into a single black hole $w$. If there are tokens on $u$ or $v$, the merged black hole receives the union of all such tokens. In other words, $\omega_{I_s'}(w) = \omega_{I_s}(u) + \omega_{I_s}(v)$ and $\omega_{I_t'}(w) = \omega_{I_t}(u) + \omega_{I_t}(v)$. Loops and multi-edges are ignored.
 
\item Reduction rule R2 (dominated black hole rule): If there exists two black holes $u$ and $v$ such that $N(u) \subseteq N(v)$, $\omega_{I_s}(u) = 0$, and $\omega_{I_t}(u) = 0$, we delete $u$. 

\item Reduction rule R3 (absorption rule): If there exists $u,v$ such that $u$ is a black hole, $v \in N(u) \cap A$ ($v$ is a neighboring planet that could be in $I_s \cup I_t$) and  $|((I_s \cup I_t) \cap A ) \cap N[v]| \leq 1$, then we contract the edge $uv$. We say that $v$ is \emph{absorbed} by $u$. If $v \in I_s \cup I_t$ then we update the weights of $u$ accordingly.
 
\item Reduction rule R4 (twin planets rule): Let $u,v \in A(G)$ be two planet vertices that are \emph{twins} (true or false twins). That is, either $uv \not\in E(G)$ and $N(u) = N(v)$ or $uv \in E(G)$ and $N[u] = N[v]$. 
If $u \not\in I_s \cup I_t$ then delete $u$. 
If both $u$ and $v$ are in $I_s$ (resp. $I_t$) and at least one of them is not in $I_t$ (resp. $I_s$) then return a trivial no-instance. If both $u$ and $v$ are in $I_s$ as well as $I_t$ then delete $N[u] \cup N[v]$, decrease $k$ by two, and set $I_s' = I_s \setminus \{u,v\}$ and $I_t' = I_t \setminus \{u,v\}$. 

\item Reduction rule R5 (path reduction rule): Let $G$ be a galactic graph and $P$ be a $A$-geodesic path of length at least $5k$ such that  $(A \cap N[P]) \cap (I_s \cup I_t) = \emptyset$. Then, $P$ can be contracted into a black hole (we ignore loops and multi-edges). That is, we contract all edges in $P$ until one vertex remains. 
\end{itemize}

Note that all of the above rules allow us to reduce the size of the input graph.
In the remainder of this section, we prove a series of lemmas establishing the safety of the aforementioned rules. We apply the rules (or a subset of them) in order. That is, every time a rule applies, we start again from the first rule. 
This way, we assume that a rule is applied exhaustively before moving on to the next one.

\begin{lemma}\label{lem:easy_direction}
Let $(G = (A \cup B, E),k,I_s,I_t)$ be an instance of {\sc Galactic Token Sliding} and let $Q$ be any subset of $V(G)$.
Let $(G',k,I_s',I_t')$ be the instance obtained by identifying all the vertices of $Q$ into a single black hole vertex $q$ which is adjacent to every vertex in $N_G(Q) \setminus Q$ (loops and multi-edges are ignored).
We set $\omega_{I_s'}(q) = |Q \cap I_s|$ and $\omega_{I_t'}(q) = |Q \cap I_t|$.
If $(G,k,I_s,I_t)$ is a yes-instance then $(G',k,I_s',I_t')$ is a yes-instance.
\end{lemma}

\begin{proof}
Assume that there exists a transformation from $I_s$ to $I_t$ in $G$. Let $\langle I_0 = I_s, I_1, \ldots, I_\ell = I_t \rangle$ be such a transformation.
To obtain a transformation in $G'$ we simply ignore all token slides that are restricted to edges in $G[Q]$.
Formally, we delete any $I_i$ that is obtained from $I_{i-1}$ by sliding a token along an edge in $G[Q]$. For every $I_i$ obtained from $I_{i-1}$ by sliding a token from $N_G(Q) \setminus Q$ onto $Q$, we instead slide the token to $q$ and increase the weight of $q$ by one, i.e., $\omega_{I'_i}(q) = \omega_{I'_{i-1}}(q) + 1$.
We replace every $I_i$ obtained from $I_{i-1}$ by sliding a token from
$Q$ to $N_G(Q) \setminus Q$ by $I'_{i-1}$ and $I'_{i-1}$ where one token gets projected from $q$ onto its corresponding neighbor (and we decrease the weight of the black hole by one).
All other slides in the sequence are kept as is and we obtain the desired sequence $\langle I'_0 = I_s', I'_1, \ldots, I'_{\ell'} = I_t' \rangle$ from $I_s'$ to $I_t'$ in $G'$.
\end{proof}

\begin{lemma}\label{lem:rule_mergingbh}
Reduction rule R1, the adjacent black holes rule, is safe.
\end{lemma}

\begin{proof}
Let $G$ be the initial galactic graph and $G'$ be the graph obtained after contracting the two adjacent black holes $u$ and $v$ into a single black hole $b$.
Let $I_s,I_t$ be the two galactic independent sets of $G$ and let $I_s',I_t'$ be their counterparts in $G'$.
If there is a transformation from $I_s$ to $I_t$ in $G$, then, by Lemma~\ref{lem:easy_direction}, there is a transformation from $I_s'$ to $I_t'$ in $G'$.

Assume now that there is a transformation from $I_s'$ to $I_t'$ in $G'$.
We adapt it into a sequence in $G$ maintaining the fact that, at each step, the weight of every vertex $z \ne u,v$ is the weight of $z$ at the same step of
the transformation in $G'$ and $\omega(u)+\omega(v)=\omega(b)$.
Note that $I_s$ (resp.\ $I_t$) satisfies these conditions with $I_s'$ (resp.\ $I_t'$).
We perform the same sequence in $G$ if possible, that is, if both vertices exist in $G$, we perform the slide (which is possible by the above condition).
Now, let us explain how we simulate the moves between $b$ and its neighbors.
If a token on $z \in N(b)$ slides to $b$, then in $G$ we simulate this move by sliding the corresponding token to $u$ or $v$, depending on which vertex $z$ is incident to (note that if $z \in N(u) \cap N(v)$, then the token on $z$ can be slid to $u$ or $v$).
If the move corresponds to a token leaving $b$ in $G'$ to a vertex $z$, then
if a vertex in $\{u,v\}$ incident to $z$ has positive weight, we slide a token from one of these vertices to $z$.
So we can assume, up to symmetry, that $u$ is incident to $z$ and $\omega(u)=0$.
Since $\omega(u)+\omega(v)=\omega(b)$ (at every step) and a token leaves $b$ in $G'$, we have $\omega(v)>0$. Hence, we can move a token from $v$ to $u$, and eventually move this token from $u$ to $z$.
\end{proof}

\begin{lemma}\label{lem:rule_domination}
Reduction rule R2, the dominated black hole rule, is safe.
\end{lemma}

\begin{proof}
Let us denote by $G$ the original galactic graph and $G'$ the graph where $u$ has been deleted.
Clearly, every reconfiguration sequence in $G'$ is a reconfiguration sequence in $G$.
We claim that every reconfiguration sequence in $G$ from $I_s$ to $I_t$ can be adapted into a reconfiguration sequence where the dominated black hole vertex $u$ never contains a token.
Consider a reconfiguration sequence from $I_s$ to $I_t$ that minimizes the number of times a token enters $u$, and suppose for a contradiction that at least one token enters $u$.
Let $s$ be the last step where a token enters $u$ and $s'$ be the next time a token is leaving from $u$ (note that both steps $s$ and $s'$ exist, since $\omega_{I_s}(u) = \omega_{I_t}(u) = 0$).
Instead of moving a token to $u$ at step $s$ we move it to $v$ and at step $s'$, we move the token from $v$ (which is possible since $N(u) \subseteq N(v)$).
It still provides a sequence from $I_s$ to $I_t$ and the number of times a token enters $u$ is reduced, a contradiction with our choice of sequence.
Hence, there exists sequence from $I_s$ to $I_t$ in $G$ such that $u$ never contains a token, and thus the rule is safe.
\end{proof}

In what follows we assume that $|I_s \cap N(b) \cap A| \leq 1$ and $|I_t \cap N(b) \cap A| \leq 1$ for each black hole $b$. In other words, there is at most one token of the initial and target independent sets in the neighborhood of a black hole. This is a safe assumption for the following reason. 
Suppose that $N(b) \cap A$ contains at least two vertices of $I_s$ or at least two vertices of $I_t$, for some black hole $b$. Let $I_s'$ (resp. $I_t'$) be the galactic independent set obtained by moving the token on $I_s$ (resp. $I_t$) to $b$. By definition of black holes this is a valid reconfiguration sequence, and thus there is a sequence transforming $I_s$ to $I_t$ if and only if there is one transforming $I_s'$ into $I_t'$, and $N(b) \cap A$ contains no vertex of the initial nor target independent set in its neighborhood.

\begin{lemma}\label{lem:nice_seq}
Assume that there exists a sequence $\langle I_0 = I_s, I_1, \ldots, I_\ell = I_t \rangle$ between two galactic independent sets $I_s$ and $I_t$ of a galactic graph $G$ such that $|I_s \cap N(b) \cap A| \leq 1$ and $|I_t \cap N(b) \cap A| \leq 1$, for every black hole $b$. Then this sequence can be modified such that for each black hole $b$ we have at most one token on $N(b) \cap A$ at all times, i.e., $|I_i \cap N(b) \cap A| \leq 1$ for all $0 \leq i \leq \ell$.
\end{lemma}

\begin{proof}
Consider a reconfiguration sequence $\langle I_0 = I_s, I_1, \ldots, I_\ell = I_t \rangle$ from $I_s$ to $I_t$ and suppose that there exists a black hole $b$ such that, at some point in the sequence, $N(b) \cap A$ contains two tokens. 
Let $I_i$ be the first galactic independent set in the transformation with two tokens on $N(b) \cap A$. By the choice of $i$, $I_i$ is not the first nor the last independent set of the sequence.
In the previous galactic independent set $I_{i-1}$ in the sequence, there is a unique token $t$ on $N(b) \cap A$.
Let $v_t \in N(b) \cap A$ be the vertex containing $t$ in $I_{i-1}$.
Let $s$ be the step when $t$ enters $v_t$ and does not leave it until at least after $I_i$ (with $s$ possibly equal to $0$).
We add a move in the sequence just after $s$ consisting in sliding $t$ from $v_t$ to the black hole $b$. Similarly, we add a move just after $i$ consisting in sliding $t'$ (the token entering $N(b)$ at step $i$) from $v_{t'}$ to $b$. 
Note that, regardless of which tokens slides next, we can perform these slides by first projecting the corresponding token out of the black hole. 

Hence, we obtain a new reconfiguration sequence where the number of steps with at least two tokens on the neighborhood of $b$ has strictly decreased. We can repeat this argument as many times as needed on every black hole of $G$, up until we obtain a sequence from $I_s$ to $I_t$ where no black hole ever has two tokens in its neighborhood.
\end{proof}

\begin{lemma}\label{lem:rule_absorbe}
Reduction rule R3, the absorption rule, is safe.
\end{lemma}

\begin{proof}
Let $u$ be a black hole with a planet neighbor $v$ such that $|(I_s \cup I_t) \cap A \cap N[v]| \leq 1$. 
We denote by $G'$ the galactic graph where $u$ and $v$ are contracted into black hole $b$ and we let $I_s'$ and $I_t'$ be the galactic independent sets corresponding to $I_s$ and $I_t$. If there is a transformation from $I_s$ to $I_t$ in $G$, then, by Lemma~\ref{lem:easy_direction}, there is a transformation from $I_s'$ to $I_t'$ in $G'$.
Consider a transformation from $I_s'$ to $I_s'$ in $G'$.
We claim that the transformation in $G'$ can be changed into a transformation in $G$.
 
By Lemma~\ref{lem:nice_seq}, we can assume the existence of a sequence in $G'$ where the number of tokens in $N(b) \cap A$ is at most one throughout the sequence.
If there is a move in $G'$ between two vertices $x$ and $y$ where $x,y \notin N(b)$, then the same move can be performed in $G$.
If a token $t$ in the sequence in $G'$ has to move to $N(b)$ from a position distinct from $b$, then we first move the token $t'$ on $v$ (if such a token exists) to $u$ in $G$ (since $N_{G}(v) \setminus \{u\} \subseteq N_{G'}(b)$, leaving a token on $v$ in $G$ may result in two tokens being adjacent) before moving $t$.
So we are left with the case where a token has to enter or leave $b$.
If the token enters $b$ from a neighbor $w$ of $u$ (in $G$), then we simply move the token to $u$ (in $G$).
So we can assume that the token enters $b$ from a neighbor $w$ of $v$ (in $G$).
In that case, we can perform the slides $w$ to $v$ and then $v$ to $u$ to put the token on the black hole.
Such a transformation is possible since there is no other token on $N(v)$ (in $G'$, there is at most one token in $N(b)$ at all times).
Similarly, if a token has to go to some vertex $w$ of $N(v)$ from $b$, then there is currently no token on $N(v)$, and thus the sequence of moves $u$ to $v$ and $v$ to $w$ is possible, which completes the proof. 
\end{proof}

As immediate consequences, the following properties hold in an instance where reduction rules R1, R2, and R3 cannot be applied.

\begin{corollary}\label{cor:planet-neighbor}
Each (planet) neighbor of a black hole must have at least two vertices of $I_s \cup I_t$ in its planet neighborhood.
\end{corollary}

\begin{proof}
If a planet neighbor of a black hole has at most one vertex of $I_s \cup I_t$ in its planet neighborhood, then reduction rule R3 can be applied and we get a contradiction.
\end{proof}

\begin{corollary}\label{cor-planet-components}
Every planetary component must contain at least one token and therefore $G$ can have at most $k$ planetary components, when $k \geq 2$.
\end{corollary}

\begin{proof}
Assume that some planetary component $C$ contains zero tokens. Since we always assume the input graph to be connected (and none of the reduction rules disconnect the graph), all vertices of
the component will be absorbed by neighboring black holes (by reduction rule R3).
\end{proof}

\begin{lemma}\label{lem:rule_twins}
Reduction rule R4, the twin planets rule, is safe.
\end{lemma}

\begin{proof}
Let $u,v \in A(G)$ be two planet vertices such that either $uv \not\in E(G)$ and $N(u) = N(v)$ or $uv \in E(G)$ and $N[u] = N[v]$. 

Assume $u \not\in I_s \cup I_t$. 
Note that in any transformation from $I_s$ to $I_t$ we can have at most one token in $\{u,v\}$; once a token is on $u$ or $v$, the neighborhood of the other vertex cannot contain a token. Hence deleting $u$ is safe, as we can use $v$ instead. 

If both $u$ and $v$ are in $I_s$ (resp. $I_t$) and at least one of them is not in $I_t$ (resp. $I_s$) then since none of these tokens can ever slide out of $u$ or $v$ we can safely return a trivial no-instance. 

Finally, if both $u$ and $v$ are in $I_s$ as well as $I_t$ then, since those tokens can never slide, we can safely delete $N[u] \cup N[v]$, decrease $k$ by two, and set $I_s' = I_s \setminus \{u,v\}$ and $I_t' = I_t \setminus \{u,v\}$. 
\end{proof}

\begin{lemma}\label{lem:rule_mergingpaths}
Reduction rule R5, the path reduction rule, is safe.
\end{lemma}

\begin{proof}
Let $P$ be an $A$-geodesic path of length $5k$ in $G$ such that no vertex of $A \cap N[P]$ are in the initial or target independent sets, $I_s$ and $I_t$.
Let $G'$ be the graph obtained after contracting $P$ into a single black hole $b$ (recall that multi-edges and loops are ignored).
Let $I_s'$ and $I_t'$ be the galactic independent sets corresponding to $I_s$ and $I_t$.
If there is a transformation from $I_s$ to $I_t$ in $G$, then, by Lemma~\ref{lem:easy_direction}, there is a transformation from $I_s'$ to $I_t'$ in $G'$.
We now consider a transformation from $I_s'$ to $I_t'$ in $G'$ and show how to adapt it in $G$.

By Lemma~\ref{lem:nice_seq}, we can assume the existence of a sequence in $G'$ where the number of tokens in $N(b) \cap A$ is at most one
throughout the sequence, for any black hole $b$. If there is a slide from a vertex $u$ to a vertex $v$ in $G'$ such that $u,v \not\in N[b]$, then the same
slide can be applied in $G$. Whenever a token slides (in $G'$) to a vertex $u$ in $N(b)$, then we know that either $u$ later slides to $b$ or slides out of $N(b)$
(since we have at most one token in the neighborhood of black holes at all times).
If the token does not enter $b$, then the same slide can be applied in $G$.
If the token enters $b$, then we slide the token to a corresponding vertex in $P$ (in $G$). Following that slide, two things can happen. Either this token leaves $b$, in
which case we can easily adapt the sequence in $G$ by sliding along the path $P$. In the other case, more tokens can slide into $b$, which is the problematic case.
Note, however, that $P$ is of length $5k$ and is $A$-geodesic. Hence, every vertex $a \in A$ has at most three neighbors in $P$ and any independent set of
size at most $k$ in $A$ has at most $3k$ neighbors in $P$. This leaves $2k$ vertices on $P$ which we can use to hold as many as $k$ tokens that need to slide into $b$ (in $G'$).
In other words, whenever more than one token slides into $b$ in $G'$, we simulate this by sliding the tokens in $P$ onto the $2k$ vertices of $P$ that are free.
Since initially $(A \cap N[P]) \cap (I_s \cup I_t) = \emptyset$, every time a token enters into a vertex $v \in N(b)$ in $G'$, in $G$ we can rearrange the tokens on $P$
to guarantee that $N[v]$ contains no tokens. Finally, when a token leaves $b$ to some vertex $v \in N(b)$ (in $G'$), then we rearrange the tokens on $P$
so that a single token in $P$ becomes closest to $v$. This token can safely slide from $P$ to $v$. 
\end{proof}

\begin{corollary}\label{cor-diam}
Let $(G,k,I_s,I_t)$ be an instance of {\sc Galactic Token Sliding} where reduction rules R1, R3, and R5 (adjacent black holes rule, absorption rule, and the path reduction rule) have been exhaustively applied. Then, the graph $G$ has diameter at most $O(k^2)$. Moreover, any planetary component has diameter at most $O(k^2)$.
\end{corollary}

\begin{proof}
Suppose for a contradiction that there exists a geodesic path $P$ (not necessarily a planetary path) such that $|P| > 25k^2$. Since the path reduction rule and the adjacent black holes rules have been applied exhaustively, every maximal planetary subpath of $P$ has length at most $5k$ and no two consecutive vertices of $P$ are black holes. It follows that $P$ contains at least $4k+1$ disjoint maximal planetary subpath, each of which is adjacent to at least one black hole of $P$. Since $A$ is geodesic, every vertex of $I_s \cup I_t$ is adjacent to at most two planetary subpath of $P$. Since furthermore $|I_s \cup I_s| \leq 2k$, there exists a maximal planetary subpath $P'$ of $P$ such that $N[P'] \cap (I_s \cup I_t) = \emptyset$. Therefore there exists an edge $(b,u)$ of $E(G')$ such that $b \in P$ is a black hole and $u \in P'$ on which the absorption rule applies, a contradiction.

Assume by contradiction that the diameter of any planetary component is at least $5k(k + 1)$.
Let $P$ be a shortest $A$-path between two vertices at $A$-distance $5k(k + 1)$. Note that $P$ is $A$-geodesic.
Since the size of each independent set is at most $k$ and each planet can see at most three vertices on an
$A$-geodesic path, there is a subpath of $P$ of length at least $5k$ that does not have any vertex of the
independent set in its neighborhood. Hence, the path reduction rule can be applied, a contradiction.
\end{proof}

We now show how the galactic reconfiguration framework combined with the previous reduction rules immediately implies that \textsc{Token Sliding} is fixed-parameter tractable for parameter $k + \Delta(G)$. 

\begin{theorem}\label{thm:ts-bounded-degree-fpt}
\textsc{Token Sliding} is fixed-parameter tractable when parameterized by $k + \Delta(G)$. Moreover, the problem admits a bikernel\footnote{A kernel where the resulting instance is not an instance of the same problem.} with $k\Delta(G)^{O(k^2)} + (2k + 2k\Delta(G))\Delta(G)$ vertices. 
\end{theorem}

\begin{proof}
Let $(G, k, I_s, I_t)$ be an instance of \textsc{Token Sliding}. We first transform it to an instance of \textsc{Galactic Token Sliding} where all vertices are planetary vertices. We then apply all of the reduction rules R1 to R5 exhaustively. By a slight abuse of notation we let $(G, k, I_s, I_t)$ denote the irreducible instance of \textsc{Galactic Token Sliding}.

The total number of planetary components in $G$ is at most $k$ by Corollary~\ref{cor-planet-components} and the diameter of each such component is at most $O(k^2)$ by Corollary~\ref{cor-diam}. Hence the total number of planet vertices is at most $k\Delta(G)^{O(k^2)}$. 

To bound the total number of black holes, it suffices to note that no black hole can have a neighbor in $B \cup (A \setminus N[I_s \cup I_t])$. In other words, no black hole can be adjacent to another black hole (since the adjacent black holes reduction rule would apply) and no black hole can be adjacent to a planet without neighboring tokens (otherwise the absorption reduction rule would apply). Hence, combined with the fact that each black hole must have degree at least one, the total number of black holes is at most $(2k + 2k\Delta(G))\Delta(G)$.
\end{proof}

Theorem~\ref{thm:ts-bounded-degree-fpt} immediately implies  positive results for graphs of bounded bandwidth/bucketwidth. 
The bandwidth of a graph is the minimum over all assignments $f: V(G) \rightarrow \mathbb{N}$ of the quantity $\mathrm{max}_{uv \in E(G)} |f(u) - f(v)|$. A graph of bandwidth $b$ can easily
be seen to have pathwidth and treewidth at most $b$ and maximum degree at most $2b$. On the other hand, the family of stars $K_{1,n}$ gives an example with bounded pathwidth but unbounded bandwidth. A bucket arrangement of a graph is a partition of the vertex set into a sequence of buckets, such that the endpoints of any edge are either in one bucket or in two consecutive buckets. If a graph has a bucket arrangement where each bucket has at most $b$ vertices, then it has bandwidth at most $2b$ (arrange one bucket after another, with any ordering within one bucket).

\section{The multi-component reduction rule  (R6)}\label{sec-components}

\subsection{General idea}
The goal of this section is to show how we can reduce a graph when we have a small vertex separator with many components attached to it. We let $X$ be a subset of vertices and $H$ be an induced subgraph of $G - X$ (we assume $G$ is a non-galactic graph in this section). Let $I_{s}$ and $I_{t}$ be two independent sets which are disjoint from $H$ and consider a reconfiguration sequence from $I_s$ to $I_t$ in $G$. Let $v$ be a vertex of $H$ and assume that there is a token $t$ that is \emph{projected} on $v$ at some point of the reconfiguration sequence, meaning that the token $t$ is moved from a vertex of $X$ to $v$.
This token may stay a few steps on $v$, move to some other vertex $w$ of $H$, and so on until it eventually goes back to $X$. Let this sequence of vertices (allowing duplicate consecutive vertices) be denoted by $v_1=v,v_2,\ldots,v_r$. We call this sequence the \emph{journey} of $v$ (formal definitions given in the next section). 

Assume now that the number of connected components attached to $X$ is arbitrarily large. Our goal is to show that one of those components can be safely deleted, that is, without compromising the existence of a reconfiguration sequence if one already exists.
Suppose that we decide to delete the component $H$.
The transformation from $I_s$ to $I_t$ does not exist anymore since, in the reconfiguration sequence, the token $t$ was projected on $v \in V(H)$.
But we can ask the following question:
Is it possible to simulate the journey of $v$ in another connected component of $G - X$? 
In fact, if we are able to find a vertex $w$ in a connected component $H' \neq H$ of $G - X$ and a sequence $w_1=w,\ldots,w_r$ of vertices such that $w_iw_{i+1}$ is an edge for every $i$ and such that $N(w_i) \cap X = N(v_i) \cap X$, then we could project the token $t$ on $w$ instead of $v$ and perform this journey instead of the original journey of $t$.\footnote{We assume for simplicity in this outline that the component of $w$ does not contain tokens.} 
One possible issue is that the number $r$ of (distinct) vertices in the journey can be arbitrarily large, and thus the existence of $w$ and $H'$ is not guaranteed \textit{a priori}. 
This raises more questions:
What is really important in the sequence $v=v_1,v_2,\ldots,v_r$?
Why do we go from $v_1$ to $v_r$?
Why so many steps in the journey if $r$ is large?
The answers, however, are not necessarily unique.
We distinguish two cases.

First, suppose that in the reconfiguration sequence, the token $t$ was projected from $X$ to $v$, performed the journey without having to ``wait'' at any step (so no duplicate consecutive vertices in the journey), and then was moved to a vertex $x' \in X$. Then, the journey only needs to ``avoid'' the neighbors of the vertices in $X$ that contain a token.
Let us denote by $s_1$ the step where the token $t$ is projected on $v$ and by $s_2$ the last step of the journey (that is, the step where $t$ is one move/slide away from $X$).
Let $Y$ be the vertices of $X$ that contain a token between the steps $s_1$ and $s_2$. 
The journey of $t$ can then be summarized as follows: a vertex whose neighborhood in $X$ is equal to $N(v) \cap X$, a walk whose vertices all belong to $H$ and are only adjacent to subsets of $X \setminus Y$, and then a vertex whose neighborhood in $X$ is equal to $N(v_r) \cap X$. 
In particular, if we can find, in another connected component of $G-X$, a vertex $w$ for which such a journey (with respect to the neighborhood in $X$) also exists, then the we can project $t$ on $w$ instead of $v$.
Clearly, the obtained reconfiguration sequence would also be feasible (assuming again no other tokens in the component of $w$).

However, we might not be able to go ``directly'' from $v_1=v$ to $v_r$. Indeed, at some point in the sequence, there might be a vertex $v_{i_1}$ which is adjacent to a token in $X$.
This token will eventually move (since the initial journey  with $t$ in $H$ is valid), which will then allow the token $t$ to go further on the journey.
But then again, either we can reach the final vertex $v_r$ or the token $t$ will have to wait on another vertex $v_{i_2}$ for some token on $X$ to move, and so on (until the end of the journey). We say that there are \emph{conflicts} during the journey\footnote{Actually, there might exist another type of conflict we do not explain in this outline for simplicity.}.

So we can now "compress" the path as a path from $v_1$ to $v_{i_1}$, then from $v_{i_1}$ to $v_{i_2}$ (together with the neighborhood in $X$ of these paths), as we explained above.
However, we cannot yet claim that we have reduced the instance sufficiently since the number of conflicts is not known to be bounded (by a function of $k$ and/or the size of $X$). The main result of this section consists in proving that, if we consider a transformation from $I_{s}$ to $I_{t}$ that minimizes the number of moves ``related'' to $X$, then (almost) all the journeys have a ``controllable'' amount of (so-called important) conflicts. Actually, we prove that, in most of the connected components $H$ of $G - X$, we can assume that we have a "controllable" number of important conflicts for every journey on $H$ in a transformation that minimizes the number of token modifications involving $X$. The idea consists in proving that, if there are too many important conflicts during a journey of a token $t$, we could mimic the journey of $t$ on another component to reduce the number of token slides involving $X$.

Finally, we will only have to prove that if all the vertices have a controllable number of conflicts (and there are too many components), then we can safely delete a connected component of $G - X$.

\subsection{Journeys and conflicts}\label{sec-journeys-and-conflicts}

We denote a reconfiguration sequence from $I_s$ to $I_t$ by $\mathcal{R} = \langle I_0,I_1,\ldots,I_{\ell-1},I_\ell \rangle$. 
Let $X \subseteq V(G)$ and $H$ be a component in $G - X$ such that $I_s \cap V(H) = I_t \cap V(H) = \emptyset$. All along this section, we are assuming that tokens have labels just so we can keep track of them. For every token $t$, let $v_i(t)$, $0 \leq i \leq \ell$, denote the vertex on which token $t$ is at position $i$ in the reconfiguration sequence $\mathcal{R}$.

Whenever a token enters $H$ and leaves it, we say that the token makes a journey in $H$.  Let $I_i$ denote the first independent set in $\mathcal{R}$ where $v_i(t) \in V(H)$ and let $I_j$, $i \le j$, denote the first independent set after $I_i$ where $v_{j+1}(t) \not\in V(H)$. Then the \emph{journey $J$ of $t$ in $H$} is the sequence $(v_i(t), \ldots, v_j(t))$. The journey is a sequence of vertices (with multiplicity) from $H$ such that consecutive vertices are either the same or connected by an edge. We associate each journey $J$ with a walk $W$ in $H$. The \emph{walk $W$ of $t$} in $H$ is the journey of $t$ where duplicate consecutive vertices have been removed.

We say that a token is \emph{waiting at step $i$} if $v_i(t) = v_{i-1}(t)$; otherwise the token is \emph{active}. 
Given a journey $J$ and its associated walk $W$, we say that $w \in W$ is a \emph{waiting vertex} if there is a step where the vertex $w$ is a waiting vertex in the journey.
Otherwise $w$ is an \emph{active vertex} (with respect to the reconfiguration sequence). 
So we can now decompose the walk $W$ into waiting vertices and transition walks. That is, assuming the walk starts at $y$ and ends at $z$, we can write $W = y P_0 w_0 P_1 w_1 \ldots w_\ell P_\ell z$, where each $w_i$ is a waiting vertex and each $P_i$ is a \emph{transition walk} (consisting of the walk of active vertices between two consecutive waiting vertices). Note that the transition walks could be empty. 

We are interested in why a token $t$ might be waiting at some vertex $w$. In fact, we will only care about waiting vertices that we will call important waiting vertices.
Let $w_1,\ldots,w_\ell$ be the waiting vertices of the journey and, for every $i \le \ell$, let us denote by $[s_i,s'_i]$ the time interval of the reconfiguration sequence where the token $t$ is staying on the vertex $w_i$. Note that $s_i < s_i'$. Also note that, since when $t$ is active the other tokens are not moving; thus the position of any token different from $t$ is the same all along the interval $[s_i'+1,s_{i+1}]$ for every $i \le \ell-1$. 

Let $i < j \le \ell$ and let $w_i$ be a waiting vertex. 
We say that $w_j$ is the \emph{important waiting vertex after $w_i$} if $j>i$ and $j$ is the largest integer such that no vertex along the walk of token $t$ between $w_i$ (included) and $w_j$ (included) is adjacent to a token $t' \neq t$ or contains a token $t' \neq t$ between steps $s_i'$ and $s_j$ (note that the important waiting vertex after $w_i$ might be the last vertex of the sequence). Since token $t$ is active from $s_i' + 1$ to $s_{i+1}$ and is moving from $w_i$ to $w_{i+1}$ during that interval, the important waiting vertex after $w_i$ is well-defined and is at least $w_{i+1}$. Let $Q_{i,j}$ denote the walk in $H$  that the token $t$ follows to go from $w_i$ to $w_j$ (both $w_i$ and $w_j$ are included in $Q_{i,j}$). In other words, $Q_{i,j}=w_iP_{i+1}w_{i+1}\ldots P_jw_j$.
Now, note that since $w_j$ is the important waiting vertex after $w_i$ (i.e. we cannot replace $w_j$ by $w_{j+1}$), then we claim that the following holds:

\begin{claim}\label{clm:defconflict}
If $w_j$ is not the last vertex of the walk $W$, either 
\begin{enumerate}[(i)]
    \item\label{c1} there is a token on or adjacent to a vertex of $P_{j+1}w_{j+1}$ (the transition walk after $w_j$) at some step in $[s_i',s_j']$ or,
    \item\label{c2} there is a token on or adjacent to a vertex of $Q_{i,j} - w_j$ in the interval $[s_j,s_{j+1}]$.
\end{enumerate}
\end{claim}

Before explaining the claim indeed holds, let us define the notion of conflicts. Since we cannot replace $w_j$ by $w_{j+1}$, it means that, by definition, there is at least one step $s_q$ in $[s_i',s_{j+1}]$ where a token $t_q \ne t$ is adjacent to (or on a vertex) $v_q$ of $Q_{i,j+1}$. We call such a step a \emph{conflict}. We say that $(s_q,v_q,t_q)$ is the \emph{conflict triplet} associated to the conflict (for simplicity, we will mostly refer to a triplet as a conflict). 

We can now prove Claim~\ref{clm:defconflict}.
\begin{proof}
If the token is (on or) adjacent to $Q_{i,j}$, it cannot be in the interval $[s_i',s_j]$ by definition of $w_j$. So, if we are in case (ii), the conflict with $t'$ is after step $s_j$. And since $t$ is waiting on $w_j$, the conflict is indeed with a vertex in $V(Q_{i,j}) \setminus w_j$.
If we are in the case (i), the conflict can hold at any step between $s_i'$ and $s_j$. Note however that, after step $s_j'$, the token $t$ becomes active and goes from $w_j$ to $w_{j+1}$ in the sequence. So there is no token anymore in the neighborhood of $P_{j+1}w_{j+1}$ at step $s_j'$. In other words, if we have conflicts of type (i), there is a last such conflict in the interval $[s_i',s_j']$. 
\end{proof}

The conflicts of type (i) are called \emph{right conflicts} and the conflicts of type (ii) are called  \emph{left conflicts}. It might be possible that $w_j$ is the important waiting vertex because we have (several) left and right conflicts. We say that $w_j$ is a \emph{left important vertex} if there is at least one left conflict and a \emph{right important vertex} otherwise.
 
If $w_j$ is a left important vertex, we let the \emph{important conflict} $(s_\star,v_\star,t_\star)$ denote the first conflict associated with $Q_{i,j}$ between steps $s_j$ and $s_j'$, i.e., there exists no $s$ such that $s_j \leq s < s_\star  \leq s_j'$ such that there is a conflict at step $s$ with a token $t' \ne t$ which is either on $Q_{i,j}$ or incident to $Q_{i,j}$. 
Note that $v_\star$ cannot be a vertex of $Q_{i,j}$ since that would imply at least one more conflict before $s_\star$, hence $v_\star \in N(V(Q_{i,j}))$. 

If $w_j$ is a right important vertex, we let $(s_\star,v_\star,t_\star)$ denote the \emph{important conflict} associated with $P_{j+1}w_{j+1}$ between steps $s_i'$ and $s_j'$ as the last conflict associated to $P_{j+1}w_{j+1}$, i.e., there exists no $s$ such that $s_i' \leq s_\star < s \leq s_j'$ and there is a conflict $(s,v_s,t_s)$ such that $v_s$ in $P_{j+1}w_{j+1}$ or incident to $P_{j+1}w_{j+1}$. Note that $v_\star$ cannot be a vertex of $P_{j+1}w_{j+1}$ since that would imply at least one more conflict after $s_\star$, hence $v_\star \in N(V(P_{j+1}w_{j+1}))$. 
We use $\mathcal{C}(Q_{i,j + 1})[s_i', s_j']$ to denote all conflict triplets (left and right conflicts) associated with $Q_{i,j + 1}$ between steps $s_i'$ and $s_j'$.

To conclude this section, let us remark that the conflicts might be due to vertices of $H$ or vertices of $X$. In other words, for a conflict triple $(s, v, t) \in \mathcal{C}(Q_{i,j + 1})[s_i', s_j']$, $v$ is an $H$-conflict or an $X$-conflicts depending on whether $v$ is in $X$ or in $H$. In what follows we will only be interested in $X$-conflicts. The \emph{$X$-important waiting vertex after $w_i$} is $w_j$ where $j > i$ is the smallest integer such that 
$\mathcal{C}(Q_{i,j + 1})[s_i', s_j']$ contains at least one triplet 
$(s, v, t')$ where $t' \neq t$, $v \in X$, and $s_i' \leq s \leq s_j'$. Now given a journey we can define the sequence of $X$-important waiting vertices as the sequence $w'_1,\ldots,w'_r$ starting with vertex $w_1$ and such that $w'_{j+1}$ is the $X$-important waiting vertex after $w_j$. What will be important in the rest of the section is the length of this sequence, i.e., $r$. If this sequence is short (bounded by $f(k)$), then we can check if we can simulate a similar journey in other components efficiently. If the sequence is long, we will see that it implies that we can find a "better" transformation.

Since we will mostly be interested in how a journey interacts with $X$, we introduce the notion of the $X$-walk associated with journey $J$. The \emph{$X$-walk} is  written as $W^X = y P_0 w_0 P_1 w_1 \ldots w_\ell P_\ell z$, where each $w$ is an $X$-important waiting vertex and each $P$ is the walk that the token takes (note that this walk could have non-important waiting vertices) before reaching the next $X$-important waiting vertex. We call each $P$ in an $X$-walk an \emph{$X$-transition walk}.

\subsection{Types and signatures}
Let $X$ be a subset of vertices and $H$ be a component of $G - X$. An \emph{$\ell$-type} is defined as a sequence 
$I Y_1 W_1 Y_2 W_2 \ldots Y_\ell W_\ell Y_{\ell + 1} F$ such that for every $i$, $W_i$ is a (possibly empty) subset of $X$ and $Y_i$ is a (possibly empty) subset of $X$ or a special value $\perp$ (the meaning of $\perp$ will become clear later on). We call $I$ the \emph{initial value} and $F$ the \emph{final value} and they are both non-empty subsets of vertices of $X$. 
The $0$-type is defined as $I Y_0 F$ and we allow $I$ to be equal to $F$. 
We will often represent an $\ell$-type by $(I (Y_iW_i)_{i \le \ell}Y_{\ell + 1} F)$. Note that if $X$ is bounded, then the number of $\ell$-types is bounded. More precisely, we have:

\begin{remark}\label{rem:nr-types}
The number of $\ell$-types is at most $(2^{|X|}+1)^{2(\ell+2)}$.
\end{remark}

The neighborhood of a set of vertices $S \subseteq V(H)$ in $X$ is called the \emph{$X$-trace} of $S$. 
A journey $J$ is \emph{compatible} with an $\ell$-type $I Y_1 W_1 Y_2 W_2 \ldots Y_\ell W_\ell Y_{\ell+1} F$ if it is possible to partition the $X$-walk $W$ of $J$ into $W^X = y P_0 w_0 P_1 w_1 \ldots P_\ell w_\ell P_{\ell + 1} z$ such that:
\begin{itemize}
    \item the $X$-trace of each vertex $w_i$ is $W_i$, 
    \item for every walk $P_i$ which is not empty, the $X$-trace of $P_i$ is $Y_i$, i.e., $\cup_{x \in P_i} N(x) \cap X = Y_i$ (note that we can have $Y_i = \emptyset$), 
    \item for every empty walk $P_i$, we have $Y_i=\perp$, and
    \item the $X$-trace of $y$ is $I$ and the $X$-trace of $z$ is $F$. 
\end{itemize}
The \emph{$\ell$-signature} of a vertex $v \in V(H)$ (with respect to $X$) is the set of all $\ell'$-types with $\ell' \le \ell$ that can be \emph{simulated} by $v$ in $H$. In other words, for every $\ell$-type, there exists a walk $W$ starting at $v$ such that $W = v P_0 w_0 P_1 w_1 \ldots P_\ell w_\ell P_{\ell + 1} z$ is compatible with the $\ell$-type if and only if the $\ell$-type is in the signature. We say that two vertices are \emph{$\ell$-equivalent} if their $\ell$-signatures are the same.

\begin{lemma}\label{compute-ell-signature-fpt}
One can compute in $O^*((2^{|X|}+1)^{2(\ell+2)})$ the $\ell$-signature of a vertex $v$ in $H$.
\end{lemma}

\begin{proof}
In order to prove the lemma, we need the following simple claim.
Let $T_1,T_2,T_3$ be three types. Let us simply prove that, given a subset $A$ of vertices of $H$ all of type $T_1$, the set of vertices $B$ of type $T_3$ which can be reached via a walk whose union of types is exactly $T_2$ can be found in polynomial time. Indeed, we delete all the vertices whose type is not included in $T_2$. For each connected component $C$, if the union of the types of the vertices of the connected component are not equal to $T_2$ then a walk (whose union type is $T_2$) from a vertex of $A$ to a vertex of $B$ passing through $C$. So we can remove $C$. Now for every vertex $v$ of type $T_3$, if there is a component $C$ that is connected (or contains) a vertex of $A$ and is connected (or contains) $v$, then $v \in B$.

Now consider an $\ell$-type  $I Y_1 W_1 Y_2 W_2 \ldots Y_\ell W_\ell Y_{\ell+1} F$. We will apply the previous claim iteratively starting with $A=\{ v\}$ and setting (at every step $i$, $T_1=W_i$, $T_2=Y_{i+1}$ and $T_3=W_{i+1}$ (with $W_0=I$ and $W_{\ell+1}=F$). 
Since there are at most $(2^{|X|}+1)^{2(\ell+2)}$ $\ell$-types by Remark~\ref{rem:nr-types}, the conclusion follows.
\end{proof}

\subsection{$X$-reduced sequences and equivalent journeys}

Let $J$ be a journey and $W^X$ be the $X$-walk associated with it.
One can wonder what is really important when we consider a journey and its associated $X$-walk. Clearly, there is something special about  $X$-important waiting vertices and the conflict triples that happen before sliding to the next $X$-important waiting vertex. But what is really important in those $X$-transition walks? The only purpose of these walks (with respect to $X$) is basically to link the $i$-th $X$-important waiting vertex to the $(i+1)$-th $X$-important waiting vertex. But we cannot say that the walk is completely irrelevant since we cannot select any walk to connect these two vertices. 
Indeed, there might be some vertices of a walk whose neighborhood in $X$ actually contains a vertex on which there is a token (or there might be $H$-conflicts).
By definition of an $X$-transition walk, the neighborhood of the walk between two $X$-important waiting vertices in $X$ must be empty of tokens for the transition to happen. So (assuming no $H$-conflicts) any walk having the same $X$-trace would be ``equivalent'' to the considered walk.
In other words, if we have another walk from $w_i$ to $w_{i+1}$ which avoids the same subset of vertices in $X$ we can replace the current $X$-transition walk by it (again assuming no $H$-conflicts).

Let $J$ be a journey with exactly $r$ $X$-important waiting vertices (in its $X$-walk). Let $w_i$ and $P_i$ be respectively the $i$-th $X$-important waiting vertex and the $i$-th $X$-transition walk. Let $P_{r+1}$ be the final $X$-transition walk. Let us denote by $W_i$ the neighborhoods of $w_i$ in $X$, and by $Y_i$ the neighborhood of $P_i$ in $X$.
Let $I$ and $F$ be the neighborhoods of the initial and final vertices of the walk associated with $J$, respectively.
The \emph{type $T$ of the journey $J$} is $I (Y_i W_i)_{i \le r} Y_{r+1} F$. 

\begin{definition}
Two journeys are \emph{$X$-equivalent} whenever the following conditions are true.
\begin{itemize}
\item They have the same number of $X$-important waiting vertices;
\item The initial and final vertex of the $X$-walk have the same $X$-trace; 
\item For every $i$, the $X$-trace of the $i$th $X$-important waiting vertex is the same in both journeys; and 
\item For every $i$, the $X$-trace of the $i$th $X$-transition walk is the same in both journeys.
\end{itemize}
\end{definition}

We conclude this section by introducing the notion of $X$-reduced transformations. 
Let $I,J$ be two independent sets and $X$ be a subset of vertices of $G$. A slide of a token \emph{is related to $X$} if the token is moving from or to a vertex in $X$ (possibly from some other vertex in $X$). We call such a move an \emph{$X$-move}. 

\begin{definition}
A transformation $\mathcal{R}$ from $I$ to $J$ is \emph{$X$-reduced} if the number of $X$-moves is minimized and, among the transformations that minimize the number of $X$-moves, $\mathcal{R}$ minimizes the total number of moves. 
\end{definition}

\subsection{The multi-component reduction}

Let $H$ be a connected component of $G - X$. The $\ell$-signature of $H$ is the union of the $\ell$-signatures of the vertices in $H$. Let $\mathcal{H}$ be a subset of connected components of $G - X$. We say that $H \in \mathcal{H}$ is \emph{$\ell$-dangerous for $\mathcal{H}$} if there is a $\ell$-type in the $\ell$-signature of $H$ that appears in at most $\ell$ connected components of $\mathcal{H}$. Otherwise we say that $H$ is \emph{$\ell$-safe}. If there are no $\ell$-dangerous components, we say that $\mathcal{H}$ is $\ell$-safe.

\begin{lemma}\label{lem:safecomp}
Let $\ell = 5|X|k$. If there are more than $\ell(2^{|X|}+1)^{2(\ell+2)} + 2k + 1$ components in $G - X$, then there exists a collection of at least $2k+1$ components that are $\ell$-safe which can be found in $f(k,|X|) \cdot n^{O(1)}$-time, for some computable function $f$.
\end{lemma}

\begin{proof}
Let $\mathcal{H}_i$ be a subset of connected components of $G - X$.
We say that $H \in \mathcal{H}_i$ is \emph{$\ell$-dangerous at depth $i$} if there is an $\ell$-type in the $\ell$-signature of $H$ that appears in at most $5|X|k$ connected components of $\mathcal{H}_i$. 
Let $\mathcal{H}_{i+1}$ be the components of $\mathcal{H}_i$ which are not $\ell$-dangerous at depth $i$. 
If $\mathcal{H}_i=\mathcal{H}_{i+1}$, then all the components in $\mathcal{H}_i$ are $\ell$-safe.

So we need to prove that if the set of components $\mathcal{H}_0$ of $G - X$ is large enough, then there exists $i$ such that $|\mathcal{H}_i| \geq 2k+1$ and all the components in $\mathcal{H}_i$ are $\ell$-safe. By Remark~\ref{rem:nr-types}, there are at most $(2^{|X|}+1)^{2(\ell+2)}$ $\ell$-types.
If a component $H$ is deleted at some step $j$ (that is, $H$ belongs to $\mathcal H_{j-1}$ but not to $\mathcal H_j$), then this is because there exists some $\ell$-type $T$ that appears at most $5|X|k$ times in $H_{j-1}$ and belongs to the $\ell$-signature of $H$.
Note that all the components containing $T$ in their $\ell$-signatures are removed together at step $j$, and there are at most $5|X|k$ of them.
Hence, after at most $(2^{|X|}+1)^{2(\ell+2)}$ steps before we obtain a step $i$ such that $\mathcal{H}_i = \mathcal{H}_{i+1}$.
Therefore, less than $5|X|k(2^{|X|}+1)^{2(\ell+2)}$ components have been deleted between $\mathcal{H}_0$ and $\mathcal{H}_i$, and thus $\mathcal{H}_i$ contains at least $2k+1$ components. Furthermore, for any $t \geq 0$,  $\mathcal{H}_{t+1}$ can be computed from $\mathcal{H}_t$ in the claimed running time by Lemma~\ref{compute-ell-signature-fpt}, which concludes the proof.
\end{proof}

\begin{lemma}\label{lemma-bound-conflicts}
Let $I_{s},I_{t}$ be two independent sets and $X$ be a subset of $V(G)$. Let $\mathcal{R}$ be an $X$-reduced transformation from $I_{s}$ to $I_{t}$.  Assume that there exists a subset $\mathcal{H}$ of at least $2k+1$ connected components of $G - X$ that is $(5|X|k)$-safe. Then, for every $C \in \mathcal{H}$, any journey on the component $C$ has at most $5|X|k-1$ $X$-important waiting vertices in its associated $X$-walk.
\end{lemma}

\begin{proof}
Assume for a contradiction that there exists a safe component $C \in \mathcal{H}$ and a journey $J$ of some token $t$ in $C$ that has at least $5|X|k$ $X$-important waiting vertices. Amongst all such journeys, select the one that reaches first its $(5|X|k)$-th $X$-important waiting vertex.
Let us denote by $I (Y_i W_i)_{i \le 5k} Y_{5k + 1}$ the type of the journey $J$ that we truncate after $Y_{5k + 1}$. And, let us denote by $v (P_i w_i)_{i \le 5k} P_{5k + 1}$ the partition of the walk into $X$-important waiting vertices and $X$-transition walks (we assume the walk starts at vertex $v \in V(C)$). 

For each $X$-important waiting vertex $w_i$, let $(q_i,x_i,t_i)$ be the important conflict associated to it.
Since there are at most $k$ labels of tokens and $|X|$ vertices in $X$, there exists a vertex $x \in X$ and a token with label $t'$ such that there exists at least $5$ waiting vertices such that the important conflict is of the form $(q,x,t')$ for some $q$. In other words, there exists $P_{i_1}, \ldots, P_{i_5}$ such that for each $P_i$ we have a triplet $(q, x, t')$ (recall that $q$ denotes the step in the reconfiguration sequence). Let us denote by $s_1,\ldots,s_5$ the steps of those  conflicts whose token label is $t'$ and whose vertex in $X$ is $x$ and they appear as the important conflict triple in five different $X$-transition walks.

Let us denote by $\mathcal{C}$ the connected components of $G - X$ that contain a token at step $s_1$ and $s_5$. Since $\mathcal{H}$ contains at least $2k+1$ components and there are $k$ vertices in the independent set, there is a component $\mathcal{C}'$ in $\mathcal{H}$ which is not in $\mathcal{C}$.
Now we show we can use $C'$ to reduce the number of $X$-moves, which leads to a contradiction since we assumed that the transformation is $X$-reduced.

Since there is a vertex $w \in V(C')$ which is $5|X|k$-equivalent to $v$ (by the definition of safe component), there is, in particular a walk $W'$ starting from $w \in V(C')$ which has type $$N_X(v) (N_X(P_i) N_X(w_i))_{i \le 5k} N_X(P_{5k + 1}).$$ We denote this walk $W'$ by $w (P'_i w'_i)_{i \le 5k} P'_{5k + 1}$.

Using this walk, we can mimic the behavior of $W$ in particular between $s_1$ and $s_5$. Now the idea of the proof consists in projecting on $W'$ the token $t'$ between $s_1$ and $s_5$ which, in turn, will permit to decrease the number of $X$-moves since $t'$ will stay on $C'$ all along. Consider the step $s_1$. If the $X$-important waiting vertex $w_{i_1}$ is a left (resp. right) important waiting vertex, then the token $t'$ is moving in (resp. out) of from a vertex $x \in X$ in the initial sequence, i.e., we are considering the $X$-conflict triple $(s_1, x, t')$.

If the important waiting vertex is a left waiting vertex, the token $t'$ has to enter $X$ at step $s_1$ since this is the first $X$-conflict. Note that such a conflict happens in the interval $[s_{i_1},s'_{i_1}]$ So we can  immediately project the token $t'$ from $x$ to a vertex of $P'_{i_1}$ in component $C'$ (which is free of tokens at this point), and then slide it along $P'_{i_1}$ up to the waiting vertex $w_{i_1}$. By minimality of $s$, no vertex of $P'_{i_1}$ is incident to a token in $X$ and since the current waiting vertex is $w_{i_1}$ and $N(w_{i_1}) \cap X =N(w'_{i_1}) \cap X$ we can safely put $t'$ on $w_{i_1}'$.

If the important waiting vertex is a right waiting vertex, the token $t'$ has to leave $X$ at step $s_1+1$ since this is the last $X$-conflict being cleared before the $X$-transition walk $P_{i_1+1}$ can happen. We instead project the token $t'$ from $x$ to a vertex of $P'_{i_1+1}$ in component $C'$ (which is free of tokens at this point). Note that such a neighbor in $P'_{i_1+1}$ exists by definition of right conflicts. Then we immediately slide token $t'$ to vertex $w'_{i_1}$. 

Now token $t'$ will move to the next $X$-important waiting vertex whenever $t$ does so, i.e, it will mimic the behavior of token $t$. In other words, for every $i$, the token $t'$ will stay one the waiting vertex $w_i'$ until the token $t$ reaches $w_i$. When it does, we move the token $t'$ along the path $P_{i+1}'$ from $w_i$ to $w_{i+1}$. Note that it is possible since, by definition of next important waiting vertex, there is no vertex of $X$ that is adjacent to a vertex of $N(w_{i}P_{i+1}w_{i+1})$ in the interval $[s_i'+1,s_{i+1}]$. And by definition of types $N(w_i) \cap X = N(w_i') \cap X$, $N(w_{i+1}) \cap X = N(w_{i+1}') \cap X$ and $N(P_{i+1}) \cap X = N(P_{i+1}') \cap X$.
We repeat this mimicking until token $t$ (and $t'$) reaches $w_{i_5 - 1}$ ($w'_{i_5 - 1}$).
At this point, we consider the conflict $(s_5, x, t')$. 
Now instead of remaining in component $C'$, token $t'$ will slide back to vertex $x$ (after all other $X$-conflicts have been cleared). 
Note that, in the resulting sequence, we have at least one less move of the token $t'$ on $X$ (while the number of other moves in $X$ is not modified). Let us call this journey $J'$. So, to conclude, we have to prove that the resulting sequence remains a valid transformation from $I_s$ to $I_t$, i.e. the set of tokens is an independent set at any step.
 
By definition of equivalent journeys, we know that if there is a (new) conflict for journey $J'$, it is not with a token in $X$. We also know that it is not with the token $t$ since this token is in the component $C$ all along the journey $J'$. By our choice of $C'$, it is not with a journey that starts before $J'$ nor ends after $J'$ neither. 
So if there is a conflict, it is because there is a journey $J''$ that starts and ends between $s_1$ and $s_5$. By our choice of $C$ and $J$, the journey $J''$ has at most $5|X|k-1$ $X$-important waiting vertices and the type of the journey is not dangerous since the component $C'$ is $\ell$-safe. Hence, we can ``recursively'' apply the same reasoning and project $J''$ in a different $\ell$-safe component. Since the number of $X$-important waiting vertices in $J''$ is strictly less than in $J'$, this procedure is guaranteed to terminate, which completes the proof.  
\end{proof}

\begin{lemma}\label{lem:reducecomp}
Let $I_{s},I_{t}$ be two independent sets and $X$ be a subset of $V(G)$. If $G - X$ contains at least $4k + 2$ $(5|X|k)$-safe components, then we can delete one of those components, say $C$, such that there is a transformation from $I_{s}$ to $I_{t}$ in $G$ if and only if there is a transformation in $G - V(C)$. 
\end{lemma}

\begin{proof}
Let $C_1$, $C_2$, $\ldots$, $C_{4k+ 2}$ denote a set of safe components, let $C_{k + 1}$ denote the component that we delete, and let $C_1$ to $C_k$ denote the components that do not contain any vertex of $I_{s} \cup I_{t}$. 

If there is a transformation from $I_{s}$ to $I_{t}$ in $G - V(C)$, then, since $G - V(C)$ is an induced subgraph of $G$, there is also a transformation in $G$. 

Assume that there is a transformation in $G$ and consider an $X$-reduced sequence. We will show how to modify this transformation so that no token enters $C_{k + 1}$ and at any point in this transformation every component $C_1$, $C_2$, $\ldots$, $C_k$ contains at most one token. From Lemma~\ref{lemma-bound-conflicts}, we know that any journey on the components $C_1$, $\ldots$, $C_k$, $C_{k + 1}$ has at most $5|X|k-1$ $X$-important waiting vertices. Moreover, we know that any journey (in one of those components) with at most $5|X|k-1$ $X$-important waiting vertices can be performed in any one of the components $C_1$, $\ldots$, $C_k$ (since the components are safe). Now consider a transformation from $I_s$ to $I_t$ in $G$ and consider the first journey which either projects a second token into some component $C_1$, $\ldots$, $C_k$ or projects a token into $C_{k + 1}$. Then by the fact that at least one of the components $C_1$, $\ldots$, $C_k$ has no tokens, we modify the transformation so that $J$ occurs in an empty component.  We repeat this procedure for every journey that violates the required property to get a new transformation in $G$ which avoids $C_{k+1}$ and never projects more than one token in any of the components $C_1$, $\ldots$, $C_k$. This completes the proof.
\end{proof}

\begin{corollary}\label{cor-bounded-bag-degree}
Given a cutset $X$, we can assume that $G - X$ has at most $5|X|k(2^{|X|}+1)^{2(5|X|k+2)} + 4k + 2 = 2^{O(|X|^2k)}$ connected components. Moreover, when the number of components is larger we can find a component to delete in  $f(k,|X|) \cdot n^{O(1)}$-time, for some computable function $f$. 
\end{corollary}

\begin{proof}
For every connected component $C$ of $G - X$, we can compute the signature of $C$ in the required time by Lemma~\ref{compute-ell-signature-fpt}. By Lemma~\ref{lem:safecomp}, we can find a collection of components that are safe in the required time.
By Lemma~\ref{lem:reducecomp}, we can safely remove one of these components, which completes the proof.
\end{proof}
\section{Planar graphs}\label{sec-planar}
This section is dedicated to showing that 
\textsc{Token sliding} on planar graphs is fixed parameter tractable when parameterized by $k$. 
The proof also relies on the reduction rules designed for the \textsc{Galactic Token Sliding} problem. We say that a galactic graph $G = (V,E)$ is planar if the underlying simple graph is planar (i.e., if the graph $G$ where all vertices are considered as planetary vertices is planar). Let $(G, k, I_s, I_t)$ be an instance of \textsc{Galactic Token Sliding} where $G$ is planar. In addition to rules R1 to R6, we design in this section a number of rules whose exhaustive application (along with rules R1 to R6) results in an equivalent instance, which we denote again by $(G, k, I_s, I_t)$, such that $G$ is planar and has maximum degree bounded by a function of $k$ (the result then follows by applying Theorem~\ref{thm:ts-bounded-degree-fpt}). We first design reduction rules to bound the number of vertex-disjoint paths between any two vertices of $G$ in Section~\ref{sec:many-paths-pairs}, and then make use of this property to bound the degree of the graph in Sections~\ref{sec:fans-combs} and~\ref{sec:planar-degree}.

\subsection{Many paths between pairs of vertices}\label{sec:many-paths-pairs}

In what follows, we always consider an arbitrary (but fixed) planar embedding of $G$. A set of vertices is \emph{empty of tokens} if it contains no vertex of $I_s \cup I_t$. A vertex of a path is \emph{internal} if it is not one of the two endpoints of the path and it is \emph{external} otherwise. A set of paths of $G$ is \emph{internally vertex disjoint} if no two paths of the set share a common internal vertex.
Let $u,v \in V(G)$ and $P, P'$ be two internally vertex-disjoint $(u,v)$-paths. Since $G$ is planar, $P \cup P'$ is a separating cycle. We denote by $G_{P, P'}$ the graph induced by $P \cup P'$ and the vertices that lie inside of $P \cup P'$. The \emph{interior} of $G_{P,P'}$ is the graph induced by the vertices of $G_{P,P'} - (P \cup P')$ and is denoted by $G^o_{P,P'}$. Let $\mathcal{P}_{P,P'}$ be a set of internally vertex-disjoint $(u,v)$-paths of $G_{P,P'}$. Two paths $P_i$ and $P_{i+1}$ in $\mathcal{P}_{P,P'}$ are \emph{consecutive} if $G_{P_i, P_{i+1}} - \{u,v\}$ is connected and if any path from an internal vertex of $P_i$ to an internal vertex of $P_{i+1}$ in $G_{P_i,P_{i+1}} - \{u,v\}$ contains no internal vertex of any other path $P_j \in \mathcal{P}_{P,P'}$ that is distinct from $P_i$ and $P_{i+1}$. Note that since $G$ is planar, a path $P \in \mathcal{P}_{P,P'}$ can be consecutive with at most two other paths of $\mathcal{P}_{P,P'}$. By a slight abuse of notation, we say that $\mathcal{P} \subseteq \mathcal{P}_{P,P'}$ is a subset of consecutive paths if there exists an ordering $P_1, \ldots, P_{|\mathcal{P}|}$ such that for every $i \leq |\mathcal{P}|-1$, $P_i$ and $P_{i+1}$ are consecutive. We denote such an ordering as a \emph{consecutive ordering}. Finally, for every $i \leq |\mathcal{P}|-1$, we refer to the graph $G_{P_i, P_{i+1}}$ as the \emph{$ith$ section} of $\mathcal{P}$ and a section is \emph{internal} if it is not the first section nor the last section of $\mathcal{P}$.

Let $u,v \in V(G)$. Note that by the multi-component reduction rule (see Section \ref{sec-components}), the number connected components of $G-\{u,v\}$ is bounded by a function of $k$. Hence we can safely assume in the remaining of this section that $G-\{u,v\}$ is connected (if not, just apply the same reasoning to show that the number of vertex-disjoint $(u,v)$-path is bounded in every connected component of $G-\{u,v\}$).

\subparagraph*{Notion of crossing paths.} 
Let $u,v \in V(G)$. Let $\mathcal{P}$ be a set of internally vertex-disjoint and consecutive $(u,v)$-paths in $G$ of size $q \geq 3$, and let $P_1, \ldots, P_q$ denote a consecutive ordering of the paths in $\mathcal{P}$. Since the paths are consecutive, there exists a path from a vertex of $P_i - \{u,v\}$ to a vertex of $P_{i+1} - \{u,v\}$ for every $i \leq q-1$. Furthermore, since $q \geq 3$, the internal vertices of such a path all belong to $G^o_{P_{i}, P_{i+1}}$. It follows that there exists a path $P_C$ from a vertex of $P_{1}-\{u,v\}$ to a vertex of $P_q-\{u,v\}$ such that every internal vertex of $P_C$ belongs to $G^o_{P_1, P_q}$. We call such a path $P_C$ a \emph{crossing path of $\mathcal{P}$}, or just a \emph{crossing path} when $\mathcal{P}$ is obvious from context. Such a set of consecutive paths along with a crossing path are illustrated in Figure \ref{fig:consecutive-paths}. Let us note that since $G$ is planar, a vertex $v$ of $P_C$ can have neighbors in at most two sections of $\mathcal{P}$. More particularly, if $v$ belongs to $G^o_{P_i, P_{i+1}}$ for some $i$, then it cannot have neighbors in $G^o_{P_{j}, P_{j+1}}$ for any $j \neq i$. Our goal is now the following: We will show that as long as there is a "large enough" set of internally vertex-disjoint and consecutive $(u,v)$-paths $\mathcal{P}$, then we can contract an edge of a crossing path of $\mathcal{P}$. Formally, we show that the following reduction rules are safe:

\begin{itemize}
    \item Reduction rule P1: Let $u,v \in V(G)$ and let $P, P'$ be two internally vertex-disjoint $(u,v)$-paths such that $G_{P,P'} - \{u,v\}$ is empty of tokens. Let $\mathcal{P}$ be a set of internally vertex-disjoint and consecutive $(u,v)$-paths of $G_{P,P'}$ of size at least $10k + 21$. Let $P_1, \ldots, P_{|\mathcal{P}|}$ be a consecutive ordering of $\mathcal{P}$ such that  there exists at least $5k+10$ sections containing a non-neighbor of $u$ (which is not $v$) and at least $5k+10$ sections containing a non-neighbor of $v$ (which is not $u$).
    There exists an edge of a crossing path of $\mathcal{P}$ contained in an internal section of $\mathcal{P}$ that can be safely contracted. 
    
    \item Reduction rule P2: Let $u,v \in V(G)$ and let $P, P'$ be two internally vertex-disjoint $(u,v)$-paths such that $G_{P,P'} - \{u,v\}$ is empty of tokens. Let $\mathcal{P}$ be a set of internally vertex-disjoint and consecutive $(u,v)$-paths of $G_{P,P'}$ of size at least $(10k + 21)^2$ such that $u$ is complete to $G_{P,P'}-v$. 
    There exists an edge of a crossing path of $\mathcal{P}$ contained in an internal section of $\mathcal{P}$ that can be safely contracted.
\end{itemize}

\begin{figure}
    \centering
    \includegraphics[scale=0.7]{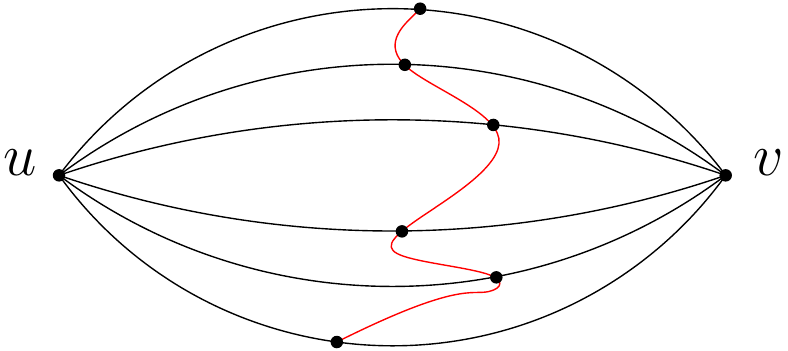}
    \caption{A set $\mathcal{P}$ of $6$ internally disjoint and consecutive $u,v$-paths defining $5$ sections. Black lines represent paths (that can be of any length). A section is a subgraph contained in the area of the plane between two consecutive paths (paths included). The red line represents a crossing path of $\mathcal{P}$. }
    \label{fig:consecutive-paths}
\end{figure}

Since reduction rules R1 to R6, P1, and P2 either delete vertices or contract edges, the graph $G$ obtained after exhaustive application of these rules is also a galactic planar graph. The remainder of this section is dedicated to proving that the rules above are safe, and then to showing that after exhaustive application of these rules the number of vertex-disjoint paths connecting any two vertices of $G$ can be bounded by a function of $k$. Note that the reduction operation is exactly the same in rule P1 and P2 and that we only divide the reduction into two separate rules for the sake of clarity. 

Let us begin by proving a couple of lemmas that will be useful to show the safety of rules P1 and P2. The following lemma shows that as long as the tokens remain "well distributed" on the interior of a set of consecutive paths, then we can always reconfigure on this interior:

\begin{lemma}\label{lem:inside-paths-reconf}
Let $\mathcal{P}$ be a set of consecutive and internally vertex-disjoint $(u,v)$-paths along with a consecutive ordering $P_1, \ldots, P_{|\mathcal{P}|}$ of these paths, and let $I$ and $J$ be independent sets of $G^o_{P,P'}$ such that $|I| = |J|$. If every two consecutive sections of $\mathcal{P}$ contain at most one vertex of $I$ and at most one vertex of $J$ then there exists a reconfiguration sequence from $I$ to $J$ in $G^o_{P,P'}$.
\end{lemma}

\begin{proof}
Let $q := |\mathcal{P}|$. If $|I \Delta J| = 0$ then we are done. Otherwise let $x \in (I \setminus J) \cap V(G_{P_i, P_{i+1}})$, for some $i \leq q$, and $y \in (J \setminus I) \cap G_{P_j, P_{j+1}}$, for some $j \leq q$. We can suppose, w.l.o.g, that $i \leq j$ and we choose $x$ and $y$ so that $j-i$ is minimum. In other words, we choose $x$ and $y$ as to minimize the number of sections that separate $x$ from $y$. Note that $x$ and $y$ are the only elements of $I \Delta J$ in $G_{P_i, P_{j+1}}$: Indeed, there can be at most one element of $I$ in $G_{P_i, P_{i+2}}$ and at most one element of $J$ in $G_{P_{j-1}, P_{j+1}}$, and by the choice of $i$ and $j$ there can be no other elements of $I \Delta J$ in $G^o_{P_{i+2}, P_{j-1}}$. Since the paths in $\mathcal{P}$ are consecutive there exists a path $P^o$ from $x$ to $y$ in $G_{P_i, P_{j+1}} - \{u,v\}$. Then there are two cases. If there is no elements of $I\cap J$ in  $G^o_{P_{i+2}, P_{j-1}}$, then the only vertices of $I \cup J$ in $N[P^o]$ are $x$ and $y$ and we can simply slide the token from $x$ to $y$ along $P^o$. Otherwise there exists $z \in I \cap J$ such that $z \in G_{P_t, P_{t+1}}$ for some $t \in [i+2, j-1]$. We choose $z$ that maximizes $t$. It follows that $z$ and $y$ are the only vertices of $I \cup J$ in $G_{t-1, j+1}$ and we can slide the token on $z$ to $y$ following a path of $G^o_{t-1, j+1}$. Hence, we can always either strictly reduce the size of $I \Delta J$ or strictly reduce the number sections that separate a vertex of $I \setminus J$ from a vertex of $J \setminus I$, which concludes the proof. 
\end{proof}

\begin{lemma}\label{ref:two-frozen-tokens}
Let $\mathcal{P}$ be a set of consecutive and internally vertex-disjoint $(u,v)$-paths of size $q$ along with a consecutive ordering $P_1, \ldots, P_{q}$ of these paths, and let $I$ be an independent set such that $\{u,v\} \subseteq I$. If $N(u) \cap G_{P_1,P_q} = N(v) \cap G_{P_1,P_q}$ then no other token can move from $G - G_{P_1,P_q}$ to $G^o_{P_1,P_q}$ as long as none of the tokens on $\{u,v\}$ move to $G - G_{P_1,P_q}$.
\end{lemma}

\begin{proof}
Let $t_u$ denote the token on $u$ and $t_v$ denote the token on $v$. As long as $t_v$ stays on $v$, the token $t_u$ cannot move to $G_{P,P'}$ and vice-versa. It is then sufficient to show that as long as $t_u$ and $t_v$ do not move, no other token can move from $G - G_{P_1,P_q}$ to $G_{P_1,P_q}$. Suppose otherwise: Since $P_1 \cup P_q$ separates $G^o_{P_1, P_q}$ from the rest of the graph, a token can only enter $G_{P_1, P_q}$ by first sliding to either $P_1$ or $P_q$. We suppose, w.l.o.g, that a token $t$ slides from $w \in G - G_{P_1,P_q}$ to $P_1$. Let $x_1, \ldots, x_\ell$ denote the vertices of $P_1$ where $x_1 = u$ and $x_\ell = v$. If $\ell = 3$ then we are done since no token can move to $x_2$. So $\ell \geq 4$ and the token $t$ must slide to some $x_i \in \{x_3, \ldots, x_{\ell-2}\}$. Since $N(u) \cap G_{P_1,P_q} = N(v) \cap G_{P_1,P_q}$ we have that $x_2 \in N(v)$ and that $x_{\ell-1} \in N(u)$ and since $G$ is planar, at most one of the edges $(v, x_2)$, $(u, x_{\ell-1})$ can be contained is the interior of $G_{P_1, P_q}$. But then the edge $(w,x_i)$ crosses the other one, a contradiction.
\end{proof}

\begin{lemma}\label{lem:rule-P1}
Reduction rule P1 is safe.
\end{lemma}

\begin{proof}
Let $q$ be the number of paths in $\mathcal{P}$ and let $P_C$ be a crossing path of $\mathcal{P}$. We pick an edge $e$ of $P_C$ that is contained in the interior of $P_3 \cup P_{q - 2}$ and contract it. Furthermore, we choose $e$ so that both endpoints of $e$ do not belong to $N(u) \setminus N(v)$ nor $N(v) \setminus N(u)$ if such an edge exists. Let $G'$ denote the contracted graph. Note that since $G_{P_1,P_q} - \{u,v\}$ is empty of tokens, both $I_s$ and $I_t$ remain independent sets of $G'$. Furthermore, if $e$ is adjacent to some path $P_i$ of $\mathcal{P}$ then we delete this path from $\mathcal{P}$. Note that the newly obtained set $\mathcal{P}$ remains a set of internally vertex-disjoint and consecutive $(u,v)$-paths for both $G$ and $G'$ which defines at least $5k+8$ sections containing a non neighbor of $u$, and the same goes for $v$.

Let us begin with a preprocessing step before we show that the reduction rule is safe. We show that up to slightly reducing the size of $\mathcal{P}$ we can always assume that $|\{u,v\} \cap I_s| \leq 1$. Suppose that there are two tokens on $\{u,v\}$. Note that initially $G_{P_1,P_q} - \{u,v\}$ is empty of tokens and hence the tokens on $\{u,v\}$ are the only tokens on $G_{P_1, P_q}$. Then there are two cases: 
\begin{itemize}
    \item $N(u) \cap V(G_{P_2, P_{q-1}}) = N(v) \cap G_{P_2, P_{q-1}}$. As long as none of the tokens on $u$ or $v$ moves during the sequence in $G$, Lemma \ref{ref:two-frozen-tokens} ensures that no other token can enter $G_{P_2, P_{q-1}}$, and thus that no token can use the edge $e$ before one of the tokens on $u$ or $v$ moves to $G - G_{P_2,P_{q-1}}$. If no such move happens in the sequence then contracting $e$ is obviously safe. Otherwise, the independent set $I'_s$ obtained after such a move satisfies $I'_s \cap \{u,v\} = 1$ and we set $\mathcal{P} := \{P_2, \ldots, P_{q-1}\}$. Note that $\mathcal{P}$ defines at least $5k+6$ sections containing a non neighbor of $u$ (resp. $v$).
    \item Otherwise there exists, w.l.o.g, a vertex $y \in N(u) \cap V(G_{P_2, P_{q-1}})$ that is not in $N(v) \cap V(G_{P_2, P_{q-1}})$. By our choice of $e$, we can choose $y$ so that it is not an endpoint of $e$ since if $e$ has an endpoint in $N(u) \setminus N(v)$ then at most one vertex of $P_C$ is not in $N(u) \setminus N(v)$ and we are free to choose $y$ so that it is not an endpoint of $e$.  Then we can slide the token on $u$ to $y$ (in both $G$ and $G'$) in some section $i \in [2, q-1]$. Since a vertex of $G_{P_1, P_{q}}-\{u,v\}$ can have neighbors in at most two consecutive sections of $\mathcal{P}$ and since $q \geq 6$, $y$ cannot have internal vertices of both $P_2$ and $P_{q-1}$ in its neighborhood. Suppose, w.l.o.g, that $y$ has no internal vertices of $P_{q-1}$ in its neighborhood. Then, we can slide the token on $v$ to an internal vertex of $P_{q-1}$ (recall that there are no other tokens but the one initially on $u,v$ in $G_{P_1, P_q})$. If $i \leq q-4$ we can slide the token on $y$ to an internal vertex of $P_2$ by following a path of $G^o_{P_2, P_{i+1}}$ (whose neighborhood is disjoint from $P_{q-1}$). We then set $\mathcal{P} := \{P_3, \ldots, P_{q-2}\}$. Otherwise, when $i \geq q-4$, we set $\mathcal{P} := \{P_2, \ldots, P_{q-5}\}$. In both case we end-up with an independent set $I's$ satisfying $I'_s \cap \{u,v\} = \emptyset$ and a set $\mathcal{P}$ which sections are empty of tokens and that defines at least $5k+1$ sections that contain a non-neighbor of $u$ (resp. $v$). 
\end{itemize}

It follows that we can always suppose that $|\{u,v\} \cap I_s| \leq 1$ and that the set $\mathcal{P}$ defines at least $5k$ sections containing a non-neighbor of $u$ and at least $5k$ sections containing a non-neighbor of $v$ (in both $G$ and $G'$). Since there is at most one token on $G_{P_1, P_q}$ initially, we are free to move this token on any vertex of $G^o_{P_1, P_q}$ different from $u$ or $v$. Moreover, we claim that we can assume that there are no black holes inside the face defined by $P_2 \cup P_{q-1}$. To see why, recall that initially $G_{P_1,P_q} - \{u,v\}$ is empty of tokens. Hence, if any black hole inside $P_2 \cup P_{q-1}$ is adjacent to a vertex in $V(G_{P_2, P_{q - 1}}) \setminus \{u,v\}$ then the absorption rule would apply (reduction rule R3). So any black hole inside $P_2 \cup P_{q-1}$ can be adjacent either to $u$ or to $v$; we cannot have a black hole adjacent to both $u$ and $v$ because of the path $P_C$. By the dominated black hole role we can have at most one black hole of each type which we can safely move to the outside of the face without modifying the rest of the embedding. 

We are now ready to prove that rule $P_1$ is safe. If there exists a sequence between $I_s$ an $I_t$ in $G'$ then there also exists one in $G$. Let us show that the other direction is also true. We consider a sequence from $I_s$ to $I_t$ in $G$ and show that there also exists one in the contracted graph $G'$. In order to do so, we consider the following independent sets of $G'$:
\begin{enumerate}
    \item An independent set $A := \{a_{i_1}, \ldots a_{i_k}\}$ such that $ A \subseteq V(G'_{P_2, P_{q-1}}) \setminus N(u)$, for every $j \leq k$ we have that $a_{i_j}$ belongs to the $i_j$-th section of $G'_{P_2, P_{q - 1}}$ and for every $j \leq k-1$, $i_j < i_{j+1}$ and $i_{j+1} - i_j \geq 3$. In other words, we pick non-neighbors of $u$ (excluding $v$ from the non-neighbors of $u$) in different sections of $G'_{P_2, P_{q-1}}$ that are ``far enough'' from each other and from the first and last sections. Such an independent set exists since $G'_{P_2, P_{q-1}}$ defines at least $5k$ sections containing a non-neighbor of $u$.
    \item An independent set $B \subseteq V(G'_{P_2, P_{q-1}}) \setminus N(v)$ defined similarly to $A$. Such an independent set exists since $G'_{P_2, P_{q-1}}$ defines at least $5k$ sections containing a non-neighbor of $v$ (excluding $u$ from the non-neighbors of $v$). 
\end{enumerate}
Note that by construction, the independent sets $A$ and $B$ satisfy the condition of Lemma~\ref{lem:inside-paths-reconf}. It follows that as long as there are no tokens on $P_2 \cup P_{q-1}$, we are free to move the tokens in between these independent sets. 
We now follow the sequence from $I_s$ to $I_t$ in $G$ and we construct a sequence in $G'$ such that there is at most one token on $F := N(V(G_{P_2,P_{q-1}})) \setminus V(G_{P_2,P_{q-1}})$ such that if a token slides on $F$, then it slides out of $F$ at the next step. We refer to the set $F$ as the \emph{frontier} of $G'_{P_2, P_{q-1}}$. Note that there is initially at most one token on $G'_{P_1,P_{q}}$, and hence we are free to move this token to a vertex of either $A$ or $B$, chosen arbitrarily (which are both included in the interior of $G'_{P_2, P_{q-1}}$, satisfying our condition). Let $\bar{G}$ denote the graph $G - N[V(G_{P_2,P_{q-1}})]$. We construct the sequence for $G'$ as follows: If a token slides between two vertices of $\bar{G}$, then the same move can be performed in $G'$, and if a token slides between two vertices of $G^o_{P_2, P_{q-1}}$ we cancel the move in $G'$ (and this process satisfies the condition on the frontier). We now have to explain how we deal with moves that include vertices of the frontier:

\begin{itemize}
    \item If a token slides from a vertex of $\bar{G}$ to $u$ (resp. $v$) in $G$, we first move the tokens on $G'_{P_2, P_{q-1}}$ to $A$ (resp. $B$) using Lemma \ref{lem:inside-paths-reconf}. We can then perform the move to $u$ (resp. $v$) in $G'$ and then further move this token to another vertex of $A$ (resp. $B$). Such a vertex $a_{i_j}$ exists since $|A| = |B| = k$ and can be reached following the internal path $P_{i_j}$ of $G'_{P_2, P_{q-1}}$. This satisfies our condition on $F$ since the only vertex of the frontier in these two moves is $u$ (or $v$). 
    
    \item If a token slides from a vertex of $\bar{G}$ to a vertex $x$ of $F \setminus \{u,v\}$ in $G$, then we  first move the tokens on $G'_{P_2, P_{q-1}}$ to $C \setminus \{c_{i_1}\}$ if $x \in N(P_2)$ or to $C \setminus \{c_{i_k}\}$ if $x \in N(P_{q-1}) \setminus N_(P_2)$. Note that by our condition on $F$, the token on $x$ is the only one on $F$ at that point. It follows that in the first case, we can further slide the token on $P_2$ and then slide to the vertex $c_{i_1}$ following a path in $G'_{P_2, P_{i_1}}$. These are valid moves, since the tokens on $G_{P_2, P_{q-1}}$ lie on the vertices $c_{i_j}$ with $j > 1$. The same goes in the second case, by first moving the token from $x$ to $P_{q-1}$ and then to $c_{i_k}$ by following a path in $G'_{P_{i_k}, P_{q-1}}$. Note that the condition on $F$ remains satisfied since all these moves are consecutive.

    \item Finally, if a token slides from a vertex of $G_{P_2, P_{q-1}}$ to $F$ in $G$, then we look ahead in the sequence: if the token stays on $F$ or slides to another vertex of $F$, or slides back to $G_{P_2, P_{q-1}}$ then we do nothing in $G'$. Otherwise it slides from $x \in F$ to a vertex of $y \in \bar{G}$ at the next step. Then in $G'$ we first slides a token from $G'_{P_2, P_{q-1}}$ to $x$ (applying Lemma \ref{lem:inside-paths-reconf} just as in the previous case if necessary), and then further slide this token to $y$. This is a valid move since the slides in $\bar{G}$ are kept as is in $G'$. Since all these slides are consecutive and since $x$ is the only vertex of $F$ involved, the condition on the frontier remains satisfied.
\end{itemize}

Note that in the constructed sequence, the tokens in $V(\bar{G})$ lie on the same vertices in $G$ and $G'$ at any step. It is in particular the case at the last step. Since furthermore $I_t \subseteq  V(\bar{G})$ the tokens all lie on $I_t$ at the last step of the constructed sequence, which concludes the proof.
\end{proof}

\begin{lemma}
Reduction rule P2 is safe.
\end{lemma}

\begin{proof}
The proof follows the exact same steps as the proof of Lemma~\ref{lem:rule-P1} with a few distinctions that we highlight next. Since there are no tokens initially in $G_{P_1,P_q} - \{u,v\}$ we know that we can assume that there are no black holes inside $P_2 \cup P_{q-1}$. We let $q$ be the number of paths in $\mathcal{P}$ and we distinguish between two cases depending on the number of sections containing a non-neighbor of $v$:

\noindent\textbf{Case 1:} There exists at most $5k + 10$ sections containing a non-neighbor of $v$. Since $q \geq (10k + 21)^2$, there must exist a subset of consecutive paths of $\mathcal{P}$ of size at least $10k + 20$ that defines sections that are also complete to $v$. We restrict $\mathcal{P}$ to the aforementioned subset and, for the sake of clarity, we again call this set $\mathcal{P}$ and denote its size by $q \geq 10k + 20$. Let $P_C$ be a crossing path of $\mathcal{P}$. We choose the edge $e$ to contract to be any edge of $P_C$ contained in the interior of $P_3 \cup P_{q - 2}$ and contract it and call $G'$ the contracted graph. We proceed to remove from $\mathcal{P}$ the paths that are adjacent to $e$, which leaves at least $10k+18$ paths in $\mathcal{P}$. Note that now we have $N(u) \cap G_{P_2, P_{q-1}} = N(v) \cap G_{P_2, P_{q-1}}$ and therefore, similarly to the proof of Lemma~\ref{lem:rule-P1}, we can assume that $|I_s \cap \{u,v\}| \leq 1$. The rest of the proof proceed by modifying the sequence in $G$ so that it can be simulated in $G'$ by restricting the number of tokens in the frontier to be at most one (just like in the proof of Lemma~\ref{lem:rule-P1}). Note that here, we can enforce the tokens on the interior of $G'_{P_2, P_{q-1}}$ to always lie on an independent set $S$ of size $k$ such that any two vertices of $S$ are separated by at least two paths of $\mathcal{P}$ (which exists since $\mathcal{P}$ contains at least $10k$ paths). This way, we can apply Lemma \ref{lem:inside-paths-reconf} to freely move the move the tokens in-between vertices of $S$ if needed. 
    
\noindent\textbf{Case 2:} There exists at least $5k + 10$ sections containing a non-neighbor of $v$. Here we can proceed exactly like in the proof of Lemma~\ref{lem:rule-P1} (including the choice of $e$) since we can easily guarantee that $|I_s \cap \{u,v\}| \leq 1$. Indeed, we can always find a vertex on $P_C \setminus N(v)$ on which to slide the token on $u$, which then allows us to slide the token on $v$ to some internal vertex of $P_2$ or $P_{q-1}$. Then, just has in the proof of Lemma~\ref{lem:rule-P1} we can also slide the token initially on $u$ onto some some internal vertex of $P_2$ or $P_{q-1}$ and restrict ourselves to considering either $\mathcal{P} := \{P_3, \ldots, P_{q-2}\}$ or $\mathcal{P} := \{P_2, \ldots, P_{q-5}\}$. We can then construct the sequence for the contracted graph just as in the proof of Lemma~\ref{lem:rule-P1}, by making use of the fact that the reduction of the size of $\mathcal{P}$ and the contraction the edge $e$ leaves more than $5k$ sections containing a non-neighbor of $v$ (and this does not change anything for $u$ which remains complete to these sections). 
\end{proof}

\begin{lemma}\label{lem-bounded-paths}
Let $u,v$ be two vertices of a galactic planar graph $G$ and $C$ be a connected component of $G-\{u,v\}$ on which rules $P1$ and $P2$ can no longer be applied. Then the number of internally vertex-disjoint $(u,v)$-paths in $C$ is at most $2k(10k+21)^3$.  
\end{lemma}

\begin{proof}
Let $\mathcal{P}$ be an inclusion-wise maximum set of internally vertex-disjoint $(u,v)$-paths in $C$ of size $q$. Suppose for a contradiction that $q > 2k(10k+21)^3$. Since $|I_s \cup I_t| \leq 2k$ there exists a subset of consecutive paths $\mathcal{P}'$ of $\mathcal{P}$ of size $r = (10k+21)^3$ such that no internal vertex of these paths belongs to $I_s \cup I_t$. Let $P_1, \ldots, P_r$ be a consecutive ordering of the paths in $\mathcal{P}'$. Since rule $P1$ does not apply we can suppose, w.l.o.g, that there exists at most $s < 5k+10$ integers $j$ such that $G_{P_j,P_{j+1}}$ contains a non-neighbor of $u$. But then there must exist a subset $\mathcal{P''}$ of consecutive paths of $\mathcal{P'}$ of size at least $\frac{r}{s} > (10k+21)^2$ such that $u$ is complete to all the sections defined by the paths in $\mathcal{P''}$. But then rule P2 would apply, a contradiction.
\end{proof}

\subsection{Reducing fans and combs}\label{sec:fans-combs}
In the previous subsection we have proved that there cannot exist too many internally vertex-disjoint paths between two vertices (by Lemma~\ref{lem-bounded-paths} combined with the multi-component reduction rule when needed). It is however not enough to bound the degree of a single vertex of a planar graph. Before proving that the degree can indeed be bounded in the next subsection, let us prove that we can reduce some particular substructures called fans and combs.

\begin{figure}
    \centering
    \includegraphics[scale=0.7]{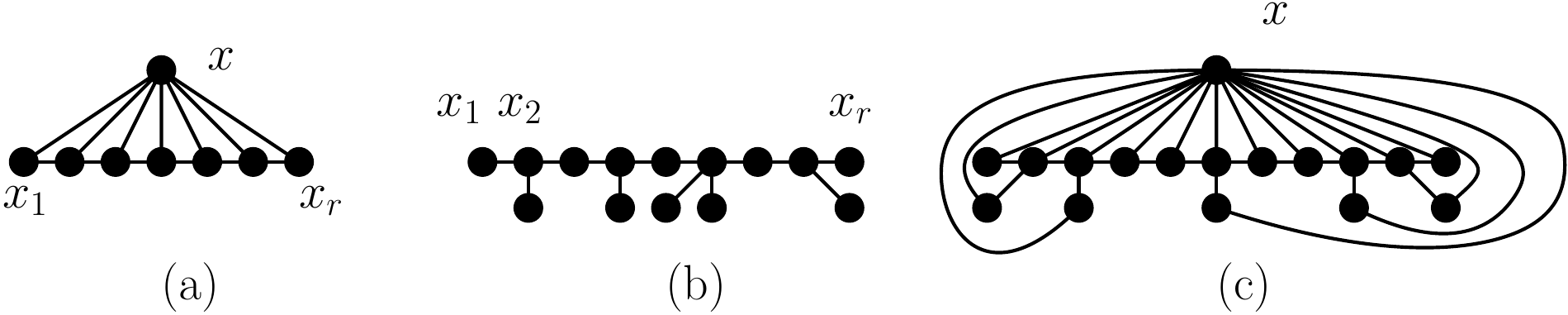}
    \caption{(a) An $r$-fan. (b) An $r$-comb (c) A complete $r$-comb. }
    \label{fig:fancomb}
\end{figure}

An \emph{$r$-fan} centered in $x$ is a graph on $r + 1$ vertices $x,x_1,\ldots,x_r$ where $x_1,\ldots,x_r$ is an induced path and $x$ is complete to $x_1,\ldots,x_r$, i.e., $x$ is adjacent to every vertex in $x_1,\ldots,x_r$. We denote by $P$ the path from $x_1$ to $x_r$ (see Figure~\ref{fig:fancomb} (a)) and we let $r$ be the \emph{size} of the fan. 
We call the \emph{interior} of an $r$-fan $I$ the set of vertices that are in the part of the plane that is enclosed by the curve $xPx$ and that does not contain the external face of the embedding. The \emph{exterior} of the fan is the set of vertices minus the interior of the fan and the vertices of the fan itself. 
We say that a fan is \emph{$I$-empty} if there is no token of the initial and target independent sets in the interior of the fan nor on the vertices of the fan itself. 
We say that a fan is \emph{$B$-empty} if there is no black hole in the interior of the fan nor on the vertices of the fan itself.
We say that an $r$-fan is \emph{safe} if no vertex but the vertices $x,x_1$ and $x_k$ have a neighbor in the exterior of the fan. We say that $x$ is \emph{$I$-complete} to the fan if $x$ is complete to all the vertices of $I$ (and to the vertices of the fan itself).
Assume that none of the previous reduction rules apply, i.e., the galactic rules and the planar rules. Then, we claim that the following reduction rule is safe:

\begin{itemize}
    \item Reduction rule P3: Let $G$ be a (galactic) graph containing a safe $r$-fan which is $I$-empty with $r \ge 3k+2$ and such that $x$ is $I$-complete to the fan. Then we can replace $I \cup \{x_2,\ldots,x_{3k+1}\}$, where $\{x_2,\ldots,x_{3k+1}\}$ denotes the induced path from $x_2$ to $x_{3k+1}$, by a path of length $3k$ whose first vertex is connected to $x_1$, the last vertex is connected to $x_{3k+2}$, and the whole path is connected to $x$.
\end{itemize}

Informally speaking, reduction rule P3 ensures that if a vertex $x$ is connected to many consecutive vertices and these vertices only have neighbors on one side of the plane (and that moreover $x$ is complete to that part of the plane), then we can replace this part by a single path. The reason this rule is safe is that we can hide as many tokens on the path as we can initially hide tokens in the interior plus $P$. Note that the new graph contains less vertices and so the rule can be applied at most $n$ times. 

\begin{lemma}\label{lem:fan}
Assume that none of the previous rules applies.
Then, reduction rule P3 is safe.
\end{lemma}

\begin{proof}
First, we show that if an $r$-fan, $r \ge 3k+2$, is safe, $I$-empty and such that $x$ is $I$-complete then the fan can be assumed to be $B$-empty. If any vertex of the fan itself is a black hole then, given that the fan is $I$-empty and safe, reduction rule R3 would apply (the absorption rule). Similarly, if any vertex in $I$ is a black hole that is adjacent to $P$ then the vertices of $P$ would be absorbed. It follows that we can only have black holes in $I$ that are adjacent to a non-empty subset of $\{x, x_1, x_r\}$. By the dominated black hole rule (reduction rule R2) and the fact that the fan is $I$-empty, we know that we can have at most one black hole of each ``type'' and each such black hole has either degree one or degree two in $\{x, x_1, x_r\}$; we cannot have a black hole in $I$ that is adjacent to all three vertices in $\{x, x_1, x_r\}$. Hence, we can modify the planar embedding of $G$ such that these black holes are no longer in $I$. So in what follows we assume that the $r$-fan is also $B$-empty. 

Let us denote by $Q$ the induced path $\{x_2\ldots x_{3k+1}\}$. Let us first recall that, since the fan is $I$-empty, $x \cup V(P) \cup I$ does not contain any tokens. 
To prove the lemma, we simply explain how we can transform an initial sequence into a new sequence where (i) the number of tokens in $V(P) \cup I$ is the same at any step and, (ii) all the tokens in $V(P) \cup I$ are actually on $P$. 
The condition holds for the initial and final step since there is no token on $V(P) \cup I$. Let us now prove it iteratively.
If the move in the initial sequence does not concern a vertex in $x \cup V(P) \cup I$ we just make exactly the same move. If we move a token of $V(P) \cup I$ to another vertex of that set, we ignore the move. 

Assume now that a token is entering in $V(P) \cup I$ via $x_1$ or $x_{3k+2}$, then we move it to $x_2$ and then $x_3$ (if it enters through $x_1$) or move it to $x_{3k+1}$ and then $x_{3k}$ (if it enters through $x_{3k+2}$). 
If there are already tokens on $V(P) \setminus \{ x_2,x_{3k+1} \}$, we move these tokens along the induced path $P$ to ensure that there is no vertex on $x_2,x_3,x_{3k-2},x_{3k}$. 
It is indeed possible due to the length of the path and since there are at most $k$ tokens on it. Similarly, if a token has to go from a vertex of $V(P)$ to $x_1$ (or $x_{3k+2}$) in the initial transformation, we move a token along the path to put a token on that vertex. 
Note that all these moves are possible and do not result in adjacent tokens. Indeed, when a token slides to some vertex of $P$ in the initial sequence there is no token on $x$ since $x$ is complete to $V(P) \cup I$. 
Moreover, since $\{x,x_1,x_{3k+2}\}$ separates $V(P) \cup I$ from the rest of the graph, we can freely move the tokens on $P$ at any point in the sequence (as long as there is no token on $x$). Finally, if a tokens enters or leaves $x$ then there is no other token on $V(P) \cup I$ in the initial sequence (since $V(P) \cup I$ is complete to $x$ and there are no black holes in $V(P) \cup I$) and thus in our reduced sequence by assumption. 
The conclusion simply follows by remarking that we can perform the same transformation if we reduce the size of $P$ by $1$ and remove the vertices of $I$ (since no token ever reaches a vertex of $I$ in the modified sequence), which completes the proof. 
\end{proof}

A case which is not handled by Lemma~\ref{lem:fan} is the case where the vertices of the path can now have neighbors on both sides of the plane separated by the fan. We prove that if we have a large enough such structure satisfying certain properties then we can also reduce it. We need a few additional definitions. 
An \emph{$r$-comb} is an induced path $P=x_1,\ldots,x_r$ of length $r$ plus possible pendant (degree-one) vertices on each vertex of the path where $x_2$ and $x_{r-1}$ have degree at least $3$ (see Figure~\ref{fig:fancomb} (b)).
A \emph{complete $r$-comb} denotes the graph consisting of an $r$-comb plus an additional vertex $x$ that is complete to the  $r$-comb (see Figure~\ref{fig:fancomb} (c)).
A complete $r$-comb is \emph{one-sided} if all the edges from $x$ to $N(P)$ (which are not in the interior of the fan $x \cup V(P)$) pass through the same side of the path (the $x_1$ side or the $x_r$ side). One can easily note that every complete $r$-comb contains a one-sided complete $r$-comb of size roughly half of the initial size (see Figure~\ref{fig:AB}).

\begin{figure}
    \centering
    \includegraphics[scale=0.6]{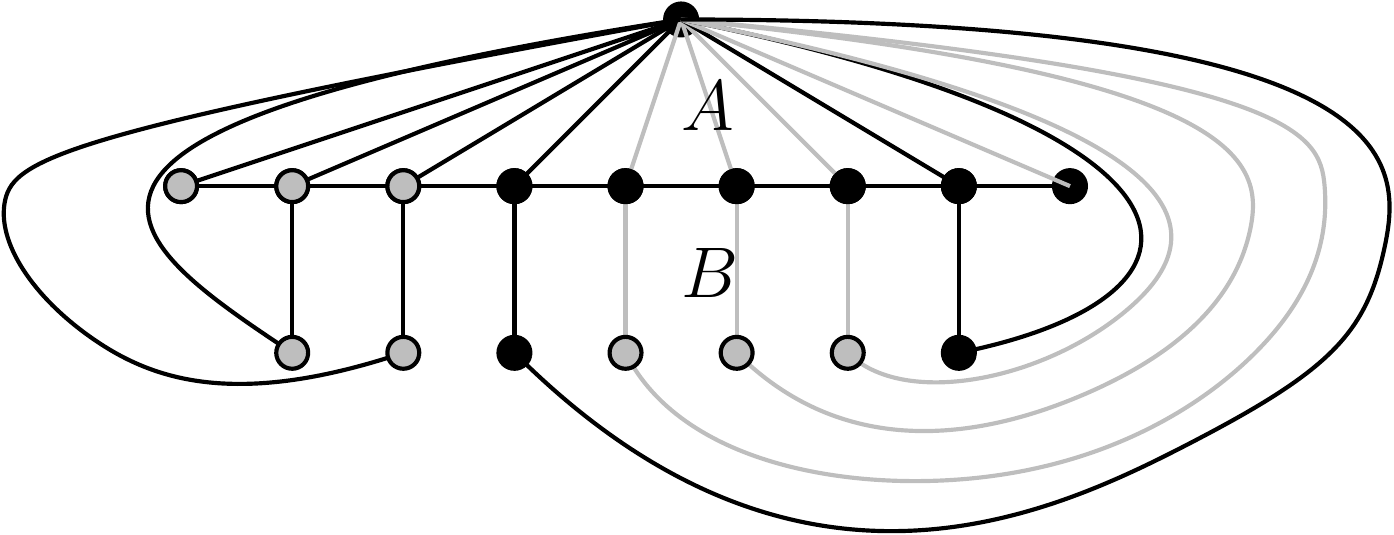}
    \caption{A one-sided complete comb in a complete comb together with the regions $A$ and $B$ of the plane.}
    \label{fig:AB}
\end{figure}

Let $C$ be a one-sided complete $r$-comb.
Let us denote by $y$ and $z$ the neighbors of respectively $x_2$ and $x_{r-1}$ that are in the same faces of respectively $x_3$ and $x_{r-2}$. Let us denote by
$H$ the graph consisting of the edges $xy$, $xz$, $x_2y$, $x_{r-1}z$, and the path $P$. The subgraph $H$ contains two faces containing $V(P) \setminus \{ x_1,x_r \}$. Let us denote by $A$ and $B$ the regions of the plane corresponding to the interior of these faces. The \emph{interior} of the $r$-comb is the set of vertices in the regions $A$ and $B$. 
We say that the complete $r$-comb is \emph{$I$-empty} if there is no token of the initial or target independent set in the interior of the $r$-comb nor on the vertices of the comb itself nor in $N(x_2)$ nor in $N(x_{r-1})$. 
We say that the complete $r$-comb is \emph{$B$-empty} if there is no black hole in the interior of the fan nor on the vertices of the comb itself.

\begin{itemize}
    \item Reduction rule P4: Assume that $G$ contains an $I$-empty one-sided complete $r$-comb with $r \ge (3k+2)^2$. Let us denote by $p_{i_1},\ldots,p_{i_d}$ the vertices of $P$ having at least one neighbor on the $B$ side. Then, we can replace $I \cup \{p_{i_1 + 1},\ldots,p_{i_d - 1}\}$, where $\{p_{i_1 + 1},\ldots,p_{i_d - 1}\}$ denotes the induced path from $p_{i_1 + 1}$ to $p_{i_d - 1}$, by a path of length at most $3k + 1$ whose first vertex is connected to $p_{i_1}$, the last vertex is connected to $p_{i_d}$ and the whole path is connected to $x$.
\end{itemize}

\begin{lemma}\label{lem:comb}
Assume that none of the previous rules applies.
Then, reduction rule P4 is safe.
\end{lemma}

\begin{proof}
Similarly to the proof of Lemma~\ref{lem:fan}, we can show that all black holes in the interior of the comb (if they exist) can be moved to the exterior. So in what follows we assume that the $r$-comb is also $B$-empty. 

Let $Q=x_3,\ldots,x_{r-2}$.
If there exists $3k+1$ consecutive vertices $Q'$ of $Q$ which only have neighbors on the $A$ part of the plane (that is on the finite part of the plane defined by the face $xQ'x$), then we can apply Lemma~\ref{lem:fan} to reduce the graph, a contradiction. 
So we can assume that for every consecutive $3k + 1$ vertices $Q'$ of the path $Q$ there is at least one neighbor of $Q'$ on the $B$ part of the plane. 
So if a subpath $Q'$ of $Q$ has length at least $(3k+2)^2$, then there exists at least $3k+2$ distinct vertices of $Q'$ that have neighbors in the $B$ part of the plane. 
Let us denote by $q_{i_1},\ldots,q_{i_d}$ the vertices of $Q'$ with a neighbor on the $B$ side. We have $d \ge 3k + 2$.  

Now we claim that we can replace $q_{i_1+1}\ldots,q_{i_d-1}$ as well as the interior of the complete $r$-comb with a path $Q_{i_1,i_d}$  of length at most $3k + 1$. The proof exactly follows the scheme of Lemma~\ref{lem:comb}. We end up with a graph where the structure between $i_1$ and $i_d$ is a $d$-fan which can be reduced to a size at most at most $3k+1$ by Lemma~\ref{lem:comb}, which completes the proof.
\end{proof}

We now have all the ingredients to reduce subdivided combs.
A \emph{subdivided complete $r$-comb} is a complete $r$-comb where every edge of the comb might be subdivided except the edges incident to the vertex $x$. When edges $x_ix_{i+1}$ of the path $P$ are subdivided we denote the subpath by $P_i$ and still require that the union of the paths $P_1,\ldots,P_{r-1} = P$ has to be an induced path. We call the path $P_i$ from $x_i$ to $x_{i+1}$ the \emph{$i$-th ray}. 

\begin{remark}\label{rem:nomberneighbors}
Assume that there does not exist more than $\ell$ internally vertex-disjoint paths between any pair of vertices of a planar graph $G$. Then:
\begin{itemize}
\item Every vertex in $G - x$ is adjacent to at most $\ell$ rays of a subdivided comb.
\item Every vertex in $G - x$ contains a vertex of at most $\ell^r$ rays of a subdivided comb in its $r$-th neighborhood (neighborhood at distance at most $r$).
\end{itemize}
\end{remark}

\begin{proof}
The proof follows from the fact that otherwise we can construct more than $\ell$ vertex disjoint paths from that vertex to $x$ which is a contradiction to our assumption.
\end{proof}

\begin{lemma}\label{lem:p-gcomb}
Let $G$ be a planar (galactic) graph and let $\ell$ be the maximum number of internally vertex-disjoint paths between any pairs of vertices of $G$. If $G$ contains a subdivided complete $r$-comb $C$, where $r \ge 18(3k+2)^2 k \cdot \ell^3$, then $G$ can be reduced.
\end{lemma}

\begin{proof}
We prove that we can find a sub-comb $C'$ of the initial subdivided comb such that:
\begin{enumerate}
    \item there is no vertex in $N^2(C')$ in the initial and target independent sets; and
    \item if there exists a sequence, there exists a sequence such that at each step there is at most one vertex in $N^2(C)$ and if a token is in $N^2(C)$ at step $i$ then there are no tokens in $N^2(C)$ at step $i-1$; and
    \item at each step, all the tokens in $C \cup N(C)$ are on the subdivided path of the comb $C'$. We will denote by $Y$ the set $C \cup N(C)$.
\end{enumerate}

Let us first prove that (1) holds for a specific sub-comb. Note that if $x$ belongs to the initial or target independent sets then we can move the token from $x$ to one of its neighbors on the subdivided comb safely. Indeed, any other vertex of the initial and target independent set shares at most $\ell-1$ neighbors with $x$ and since $x$ has at more than $2k(\ell-1)$ neighbors on the comb, there exists one vertex whose only neighbor in the starting and target independent set is $x$.

By Remark~\ref{rem:nomberneighbors}, every vertex of $I_s \cup I_t$ at distance at most $2$ from the subdivided comb has neighbors at distance $2$ in at most $\ell^2$ rays of the comb. So, since the number of rays is at least $18(3k+2)^2 k \cdot \ell^3$, there exists a sub-comb with at least $18(3k+2)^2 k \ell$ rays such that no vertex of the initial or target independent set is at distance at most $2$ from it in $G - x$. So (2) and (3) hold at the initial and final steps. Let us denote by $C'$ the resulting subdivided comb.

We cut the path of the subdivided comb $C'$ into \emph{slices} $S_i$, each slice consisting of $2(3k+2)^2$ consecutive rays. By Lemma~\ref{lem:comb} (and since every complete comb contains a one-sided complete comb of half the size), for each slice $S_i$ at least one of the following holds:
\begin{itemize}
    \item the vertex $x$ is non adjacent to a vertex of $S_i$ (i.e., one of the edges of the subdivided comb is really subdivided), or
    \item the vertex $x$ is non adjacent to a vertex of the interior of the fan $xS_ix$, or
    \item the vertex $x$ is non adjacent to a vertex of $N(S_i)$.
\end{itemize} 
The above is true because otherwise we would have a complete one-sided $(3k+2)^2$-comb on which we could apply reduction rule P4.

Since $C'$ has size at least $18(3k+2)^2 k \ell$ there are at least $3k\ell$ slices for which we have exactly the same behavior.
We now have all the ingredients to prove that there exists a transformation that satisfy (1), (2) and (3).
All along the transformation, we will keep the tokens on $Y$ in the path $P$ of the subdivided comb $C'$.

Assume that in the initial transformation some token has to slide to a vertex $z$ of $Y$. Assume first that $z \ne x$ and let $zz_1z_2$ be a path from $z$ to $C - x$. By assumption, $z$ has at most $\ell$ neighbors and so is non-adjacent to at least $3k$ different slices. So we can move tokens already in the comb in order to put them on slices not adjacent to $z_1$. Then we can simply move $z$ to $z_1$ and then $z_1$ to $z_2$.
Similarly if a token on $z_1$ is leaving $N(Y)$ in the initial transformation, we can simply move tokens on the path in order to move a token on $P$ to $z$ and then to the rest of the graph.

Assume now that a token is sliding to $x$. In each slice, there is a vertex of the slice which has a neighbor which is not a neighbor of $x$. Moreover, every vertex can be incident to at most $\ell$ slices. So, since there are at most $k$ tokens on the path, we can move them on vertices which have distinct neighbors in $N(C) \setminus N(x)$. We then move them to $N(C) \setminus N(x)$ in order to free the path of $C$ with tokens. Then, the token slides to $x$ and then is projected on the path (which is possible because of the number of neighbors of $x$ on $C$). Finally, we put back the tokens on $N(C) \setminus N(x)$ on $C$. 
A similar operation can be performed if a vertex is leaving from $x$ in the original sequence.

The conclusion simply follows by remarking that we can perform the same transformation if we reduce the size of $P$ by $1$, which completes the proof. 
\end{proof}

\subsection{Reducing high degree vertices}\label{sec:planar-degree}
It remains to prove that the previous reduction rules (galactic and planar) ensure that the (reduced) planar (galactic) graph has bounded degree. Let us first prove the following:

\begin{lemma}\label{lem:extractcomb}
Let $T$ be a tree with $r$ prescribed vertices $X$ and maximum degree at most $\ell$. Then there exists a subdivided comb of $T$ containing at least $\lfloor \log_{\ell+1}(|X|) \rfloor$ vertices of $X$. 
\end{lemma}

\begin{proof}
Let $f$ be a leaf of $T$ and let $P$ be a path defined starting from $f$ and using the following rule: If the current vertex $y$ is a leaf distinct from $f$, the path is over. Otherwise, at each step $i$ we have one or more vertices $z_i$ that we can add to $P$. We pick the $z_i$ such that the component of $T - z_i$ containing $f$ has the least number of vertices from $X$. We use $C^i_f$ to denote the component of $f$ at each step $i$.  

Note that, since $f$ is a leaf, initially $T - f$ only contains one component with at least $\ell-1$ vertices of $X$. Between steps $i$ and $i+1$, either the number of vertices of $X$ in $C^{i+1}_f$ is not modified or it decreases. Assume that it decreases. We denote by $z_i$ and $z_{i+1}$ the vertices added at steps $i$ and $i+1$, respectively. We denote by $r_i$ and $r_{i+1}$ the number of vertices of $X$ in $C^i_f$ and $C^{i+1}_f$, respectively. Finally, we let $P_{i+1}$ denote the current path. We claim that:
\begin{itemize}
    \item Either $z_i$ is in $X$ or there exists a path from $z_i$ to a vertex of $X$ in $T - P_i$; and
    \item $r_{i+1} \ge r_i/(\ell-1)-1$.
\end{itemize}
The first points follows from the fact that $X$ is decreasing so there must be a vertex of $X$ that is not in $C^{i+1}_f$. The second point is due to the the fact that the maximum degree of the tree is $\ell$ and in the worst case the vertices of $X$ are balanced (divided equally) among the components.

Now, since we started with $r$ vertices in $X$, we are guaranteed that when the procedure terminates we have a subdivided comb of logarithmic size, which completes the proof. 
\end{proof}

\begin{lemma}
Let $G$ be a galactic planar graph. If there exists a planet vertex $v \in V(G)$ whose planetary degree is greater than $(\ell+1)^{18(3k+2)^2 k \cdot \ell^3}$ then $G$ can be reduced. 
\end{lemma}

\begin{proof}
First we note that, for every vertex $v$, we can assume that there is a bounded number of connected components attached to $v$ by Corollary~\ref{cor-bounded-bag-degree}.

Assume now that there exists a vertex $v$ that has $r$ neighbors $X$ in a component $C$ of $G - v$ and let $T$ be a Steiner tree containing all the vertices of $X$. Note that the maximum degree of $T$ is bounded by $\ell:=f(k)$; as otherwise, we have a vertex $u$ of $T$ such that there exists too many paths between $u$ and $v$ in $C$. In this case we can apply Lemma~\ref{lem-bounded-paths}. 

So $T$ contains a subdivided comb of logarithmic size by Lemma~\ref{lem:extractcomb}, i.e. a subdivided comb with at least $18(3k+2)^2 k \cdot \ell^3$ rays. And then, by Lemma~\ref{lem:p-gcomb}, the graph can be reduced.
\end{proof}

\section{Chordal graphs of bounded clique size}\label{sec-chordal}
One of the main reasons why the proof of the multi-component reduction rule is so ``complicated'' is due to the fact that several tokens might belong to the cut set $X$ which makes the analysis hard (in particular, we might have several tokens on $X$ at the same time). As we shall see, once we restrict ourselves to chordal graphs of bounded clique size the proof of tractability becomes ``simpler''. We also note that the main reason why our result does not generalize to graphs of bounded treewidth is because we cannot prove the equivalent of Lemmas~\ref{lem:numberbag} and~\ref{lem-chordal-reduction-application} for such graphs. 

We start with a few definitions. A graph $G$ is chordal if no cycle of length greater than $3$ is induced. There are several related characterizations
of chordal graphs and we refer the reader to~\cite{10.1007/978-1-4613-8369-7_1} for more details.

A {\em tree decomposition} of a graph $G$ is a pair
$\mathcal{T} = (T, \chi)$, where $T$ is a tree and $\chi$ is
a mapping that assigns to each node $i \in V(T)$ a vertex subset
$B_i$ (called a {\em bag}) such that: 
\begin{itemize}
    \item $\bigcup_{i \in V(T)}{B_i} = V(G)$,
    \item for every edge $uv \in E(G)$, there exists a
node $i \in V(T)$ such that the bag $\chi(i) = B_i$ contains both $u$ and $v$, and
    \item for every $v \in V(G)$, the set $\{i \in V(T) \mid v \in B_i\}$
forms a connected subgraph (subtree) of $T$.
\end{itemize}
The {\em width} of any tree decomposition $\mathcal{T}$ is equal
to $\max_{i \in V(T)}|B_i| - 1$.
The {\em treewidth} of a graph $G$, $t = tw(G)$, is the minimum width of a tree decomposition of $G$. For any graph of treewidth $t$, we can compute a tree decomposition of width $t$ in $f(t) \cdot n^{O(1)}$ time~\cite{K94}. A tree decomposition is \emph{compact} if for every two adjacent bags $B$ and $B'$ we have $B \not\subseteq B'$ and $B' \not\subseteq B$. It is known that for every tree decomposition $\mathcal{T}$ of $G$, there exists a compact tree decomposition $\mathcal{T}'$ of $G$ such that each bag of $\mathcal{T}'$ is equal to some bag of $\mathcal{T}$. $\mathcal{T}'$ can be computed from $\mathcal{T}$ in polynomial time. A graph $G$ is chordal if and only if it has a compact tree decomposition $\mathcal{T}$ such that every bag induces a clique. We call such a decomposition a \emph{clique tree} of $G$. Moreover, if every clique of $G$ has size bounded by $\omega = \omega(G)$ then the size of the bags is also bounded by $\omega$. 

\begin{figure}
    \centering
    \includegraphics[scale=0.8]{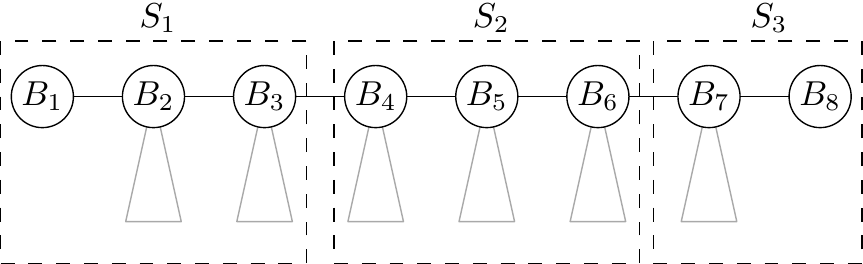}
    \caption{A tree decomposition $\mathcal{T}$ rooted at path $B_1, \ldots, B_8$ along with a $(3,2)$ block partition $\mathcal{T}$. Subtrees rooted at neighbors of nodes $B_2, \ldots, B_{7}$ are included in the rooted decomposition, whereas those rooted at neighbors of $B_1$ and $B_8$ are not.}
    \label{fig:block_partition}
\end{figure}

Given a tree decomposition $\mathcal{T}$, we let $\mathscr{P} = B_p$, $\ldots$, $B_q$ denote a path in the tree decomposition. We use $\mathcal{T}_\mathscr{P}$ to denote the \emph{tree decomposition rooted at $\mathscr{P}$}, where $\mathcal{T}_\mathscr{P}$ is obtained from $\mathcal{T}$ by deleting every subtree rooted at a neighbor of either $B_p$ or $B_q$ (except $B_{p + 1}$ and $B_{q - 1}$). Given $\mathcal{T}$, $\mathscr{P}$, and a tree decomposition rooted at $\mathscr{P}$ denoted by $\mathcal{T}_\mathscr{P}$, we define an \emph{$(\alpha,\beta)$-block partition of $\mathcal{T}_\mathscr{P}$}, denoted by $\mathcal{B}_{\alpha, \beta}(\mathcal{T}_\mathscr{P})$, as follows. We let $\mathcal{B}_{\alpha, \beta}(\mathcal{T}_\mathscr{P}) = \{\mathcal{S}_1, \mathcal{S}_2, \ldots, \mathcal{S}_\alpha\}$ where each $\mathcal{S}_i$ is a subtree (which we also call a \emph{section}) of $\mathcal{T}_\mathscr{P}$ containing at least $\beta$ distinct bags of $\mathscr{P}$. In other words, $\mathcal{S}_1$ contains the first $\frac{|\mathscr{P}|}{\alpha}$ bags of $\mathscr{P}$ along with all the subtrees attached to them except for the subtree rooted at $B_{p + \frac{|\mathscr{P}|}{\alpha}}$. $\mathcal{S}_2$ contains the bags $B_{p + \frac{|\mathscr{P}|}{\alpha}}$, $\ldots$, $B_{p + 2\frac{|\mathscr{P}|}{\alpha} - 1}$ along with all the subtrees attached to them except for the subtree rooted at $B_{p + 2\frac{|\mathscr{P}|}{\alpha}}$. We proceed as described until the last section which might contain less than $\frac{|\mathscr{P}|}{\alpha}$ bags but must contain at least $\beta$ bags of $\mathscr{P}$. We note that such a block partition is possible whenever $\mathscr{P}$ contains $\alpha \cdot \beta$ bags or more. A rooted tree decomposition along with a block partition are illustrated in Figure \ref{fig:block_partition}.

We start by showing that, in a chordal graph, one can easily modify an independent set (via token sliding) in order to be sure that every vertex of the independent set belongs to a bounded number of bags of the tree decomposition. 

\begin{lemma}\label{lem:numberbag}
Let $G$ be a chordal graph, let $\mathcal{T}$ be a compact tree decomposition/clique tree of $G$ of width at most $\omega$, and let $I$ be an independent set of $G$ of size $k$. Then, we can find a transformation from $I$ to $I'$ such that each vertex of $I'$ belongs to at most $(3\omega+3)k^2$ bags of $\mathcal{T}$.
\end{lemma}

\begin{proof}
Let $v$ be a vertex of the independent set $I$, $\mathcal{T}$ be the clique tree of $G$, and $T_v$ be the subtree of $T$ in which $v$ appears. Let $L$ denote the set of leaves in $T$, and $L_v$ denote the set of leaves in $T_v$.

If $T_v$ contains a leaf of $T$, i.e., if $L \cap L_v \neq \emptyset$, then we can slide $v$ to a vertex of this leaf bag which only appears in that bag (since the decomposition is compact). 
We now assume that $T_v$ contains no leaves of $T$, i.e., $L \cap L_v = \emptyset$, and consider the case where $T_v$ contains at least $k$ leaves, i.e., $|L_v| \geq k$. 
When $T_v$ contains at least $k$ leaves, then we can slide $v$ in such a way that it only appears on one of the leaves of $T$. Let $L_r$ be a the set of leaves of $T$ that are reachable from some leaf of $T_v$, i.e., for each leaf in $L_v$ we can associate (at least) one leaf from $L$. Indeed, since $|I| = k$ and $I$ is an independent set that contains $v$, there exists at least one leaf in $L_r$ that contains no vertex of $I$. We can then project $v$ on that branch of the tree in such a way that we end up on a leaf of $T$, as needed.  

We now assume that $T_v$ contains no leaves of $T$ and has less than $k$ leaves, i.e., $L \cap L_v = \emptyset$ and $|L_v| < k$. This implies that if $|T_v| > (3\omega+3)k^2$ then $T_v$ contains a (clique) path of length at least $3\omega+3$ such that each internal node on this path has degree two in $T$. 
Let $B_p$, $\ldots$, $B_r$, denote the  bags of this path.
So there exists a bag $B_q$, $p < q < r$, containing a vertex that is neither on the left border of the path nor on the right border of the path (neither in $B_p$ nor $B_q$). So we can simply project $v$ on that vertex that belongs to less nodes than $v$. We repeat this process as long as there exists a vertex of the independent set that belongs to more than $(3\omega+3)k^2$ bags of $T$.
\end{proof}

Given an $(G,k,I_s,I_t)$ of {\sc Token Sliding}, where $G$ is chordal and of bounded clique size $\omega$, we can assume, by Lemma~\ref{lem:numberbag}, that each vertex of both $I_s$ and $I_t$ belongs to at most $(3\omega+3)k^2$ bags of a tree decomposition $\mathcal{T}$ of $G$. By Corollary~\ref{cor-bounded-bag-degree} (the multi-component reduction rule), we can assume that each bag in the tree decomposition has at most $2^{O(\omega^2k)}$ neighbors. Moreover, if the diameter of the graph is large enough, then by reduction rule R5 (the path reduction rule), we know that there is a subset of vertices that can be merged into a black hole. Note that the resulting graph is still a chordal graph since contracting edges in a chordal graph leaves a chordal graph. So we can also assume that the diameter of $G$ is also bounded by $O(k^2)$.

To complete this section, it suffices to bound the number of bags by a function of $\omega$ and $k$ since we obtain an equivalent graph of size $f(k,\omega)$ (given that the clique number is bounded). So in the rest of this section we assume all of the above and we assume that the number of bags in the tree decomposition is arbitrarily large. Given that the diameter is bounded and the degree of the tree decomposition is bounded, it follows that when the tree decomposition is arbitrarily large we must have a long path in the decomposition in which some vertices appear in most bags (as otherwise the diameter would not be bounded). This long path can of course have subtrees attached to it but we will show that we can always find a ``big enough'' substructure that does not contain any vertices of $I_s$ nor $I_t$ (this will follow from Lemma~\ref{lem:numberbag}). To conclude, we shall show that such a substructure can be reduced using the following reduction rule (or the multi-component reduction rule R6).

\begin{itemize}
\item Rule C1 (almost black hole rule): Let $Z = W \cup X \cup Y$ be the union of three cliques of $G$ (not necessarily disjoint and one or two possibly empty) such that $Z \cap (I_s \cup I_t) = \emptyset$. Let $C$ be a connected component of $G - Z$ such that no token of $I_s$ nor $I_t$ belongs to $C$. Assume moreover that $X$ (when non-empty) can be partitioned into $X_1$ and $X_2$ such that:
\begin{itemize}
    \item $X_1$ is complete to $C \cup W \cup Y$; and
    \item $C$ is connected and there exists a $2$-independent set $R$ in $C - (N_C(W) \cup N_C(Y))$ of size at least $k$ such that: 
    \begin{itemize}
        \item for every $z$ in $W \cup Y$ there exists $r$ in $R$ such that there is a path from $r$ to $z$; and
        \item every $x$ in $X_2$ misses (is not adjacent to) at least $k$ vertices of $R$. 
    \end{itemize} 
\end{itemize}
Then, $C$ can be replaced by a path $P = \{p_1, \ldots, p_{5k}\}$ of length $5k$ such that $X_1$ is complete to $P$ and $X_2 \cup W \cup Y$ is complete to $\{p_1\}$. We note that the resulting graph is still chordal and the clique number does not increase.
\end{itemize}

\begin{lemma}\label{lem-chordal-reduction-safe}
Rule C1, the almost black hole rule, is safe.
\end{lemma}

\begin{proof}
If there exists a transformation from $I_s$ to $I_t$ in $G$ then we can find a transformation that behaves as follows:
\begin{itemize}
    \item Whenever there is a token in $X_1$ then we can have no tokens in $C \cup W \cup Y$.
    \item Whenever a token slides into $X_2 \cup W \cup Y$ from $V(G) \setminus V(C)$ it then immediately slides into a vertex of $R$ after possibly rearranging some tokens in $R$.
    \item Whenever a token slides into $X_1$ from $V(G) \setminus V(C)$ then we know that there can be no tokens in $C \cup W \cup Y$. Hence we can again project the token on any vertex of $R$. 
    \item Whenever a token slides into $X_2 \cup W \cup Y$ from $V(C)$ then the tokens in $R$ are rearranged appropriately and then the token can slide to the corresponding vertex. 
    \item Whenever a token slides into $X_1$ from $V(C)$ then we know that this must be the only token in $C$. 
\end{itemize}
The fact that we can always rearrange the tokens in a 2-independent set of size at least $k$ follows from Lemma 3.7 in~\cite{DBLP:journals/algorithmica/BartierBDLM21}. 
We call such a transformation a well-behaved transformation. 

The proof of the lemma now easily follows by considering well-behaved transformations. Assume that there exists a well-behaved transformation from $I_s$ to $I_t$ in $G$. Let $G'$ be the graph obtained after a single application of reduction rule C4 and let $I_s'$ and $I_t'$ be the corresponding independent sets in $G'$. We construct a transformation from $I_s'$ to $I_t'$ as follows:
\begin{itemize}
    \item Whenever a token slides into $X_2 \cup W \cup Y$ from $V(G) \setminus V(P)$ it then immediately slides to $p_1$ and then slides to the last available vertex of $P$. 
    \item Whenever a token slides into $X_1$ from $V(G) \setminus V(P)$ then we know that there can be no tokens in $P \cup W \cup Y$. Hence we can project the token on the last vertex of $P$. 
    \item Whenever a token slides into $X_2 \cup W \cup Y$ from $V(P)$ then the first available can slide to $p_1$ and then to the corresponding vertex in $X_2 \cup W \cup Y$. 
    \item Whenever a token slides into $X_1$ from $V(P)$ then we know that this must be the only token in $P$. 
\end{itemize}

The other direction follows immediately since a transformation from $I_s'$ to $I_t'$ in $G'$ corresponds to a well-behaved transformation in $G$. 
\end{proof}

Let us now explain how we can use Lemma~\ref{lem-chordal-reduction-safe} to bound the total number of bags in the tree decomposition. 

\begin{lemma}\label{lem-chordal-reduction-application}
Let $(G,k,I_s,I_t)$ be an instance of {\sc Galactic Token Sliding} where $G$ is chordal and of bounded clique size $\omega$ and each vertex of both $I_s$ and $I_t$ belongs to at most $(3\omega+3)k^2$ bags of a compact tree decomposition $\mathcal{T}$ of $G$ (of width at most $\omega$). 
Moreover, assume that each bag in $\mathcal{T}$ has degree at most $\gamma = 2^{O(\omega^2k)}$ and $G$ has diameter at most $O(k^2)$. Let $\mathscr{P} = B_1$, $\ldots$, $B_\ell$ denote the longest path in the tree decomposition. If $\ell > (30k^3\omega^3\gamma)^{5k\omega\gamma} \cdot 6k^3(\omega + 1) = \zeta$ then we can either apply reduction rule C1, i.e., the almost black hole rule, or reduction rule R6, i.e., the multi-component reduction rule (we note that the bound on $\ell$ is not optimized and chosen only for convenience). 
\end{lemma}

\begin{proof}
Since each vertex of $I_s \cup I_t$ belongs to at most $3k^2(\omega+1)$ bags of $T$, we can find a subpath of $\mathscr{P}$ which is disjoint from $I_s \cup I_t$ and of length at least $\ell \geq (30k^3\omega^3\gamma)^{5k\omega}$. For readability we again denote this path by $\mathscr{P} = B_1$, $\ldots$, $B_\ell$. 

Let $\mathcal{T}_\mathscr{P}$ denote the tree decomposition rooted at $\mathscr{P}$. Let $\alpha = 6k\omega\gamma$, $\beta = 5k^2\omega^2$ (so $\alpha \cdot \beta = 30k^3\omega^3\gamma$), and let $\mathcal{B}_{\alpha, \beta}(\mathcal{T}_\mathscr{P}) = \{\mathcal{S}_1, \mathcal{S}_2, \ldots, \mathcal{S}_\alpha\}$ denote an $(\alpha,\beta)$-block partition of $\mathcal{T}_\mathscr{P}$. Note that, initially, each section in $\mathcal{B}_{\alpha, \beta}(\mathcal{T}_\mathscr{P})$ is of size at least  $(5k^2\omega^2)^{5k\omega\gamma}$, i.e., contains $(5k^2\omega^2)^{5k\omega\gamma}$ bags of $\mathscr{P}$ (and possibly others). We now proceed in two stages of refinement to construct the sets $X = X_1 \cup X_2$, $W$, $Y$, $Z$, and the component $C$ required for the application of reduction rule C1 (assuming rule R6, the multi-component reduction rule, does not apply). 

If there exists a vertex $x$ that appears in all bags of $\mathscr{P} \cap \mathcal{S}$, for some $\mathcal{S} \in \mathcal{B}_{\alpha, \beta}(\mathcal{T}_\mathscr{P})$, then we add $x$ to $X$ and we restrict $\mathscr{P}$ to the subpath $\mathscr{P} \cap \mathcal{S}$ and recompute $\mathcal{T}_\mathscr{P}$ and $\mathcal{B}_{\alpha, \beta}(\mathcal{T}_\mathscr{P})$ with $\alpha = 6k\omega\gamma$ and $\beta = (5k^2\omega^2)$. Now each section is of size at least $(5k^2\omega^2)^{5k\omega\gamma - 1}$. We repeat this process as long as we can find a vertex $x$ that appears in all bags of some $\mathscr{P} \cap \mathcal{S}$. Note that this process can only be repeated at most $\omega$ times. Hence, when no longer applicable, we know that we have $5k\omega\gamma$ sections each containing at least $(5k^2\omega^2)^{4k\omega\gamma}$ bags of $\mathscr{P}$ and we have a set $X$ of vertices that appear in all bags of $\mathscr{P}$. 

We note that if $\mathscr{P} = \mathcal{T}_\mathscr{P}$ then we are done. To see why, we set $X = X_1$ ($X_2$ is empty), we set $W = B_1$, $Y = B_\ell$, and $Z = W \cup X \cup Y$. Consider the connected components of $G - Z$ that are contained in $G[\mathscr{P}]$. If we have more than $\gamma$ components then we can apply the multi-component reduction rule and we are done. Otherwise, since $\alpha = 6k\omega\gamma$, we must have at least one component that contains at least $6k\omega$ sections of $\mathcal{B}_{\alpha, \beta}(\mathcal{T}_\mathscr{P})$. Let $C$ denote this component and let $\mathscr{P}_C$ denote the bags of $\mathscr{P}$ whose vertices belong to $C$. We update $W$ to the first bag of $\mathscr{P}_C$ and $Y$ to the last bag of $\mathscr{P}_C$. To conclude, we still need to find a $2$-independent set $R$ in $C - (N_C(W) \cup N_C(Y))$ of size at least $k$. Since we have at least $6k\omega$ sections, we can find, by grouping sections into quintuplets and ignoring the first four and last four sections, such a $2$-independent set. In particular, we let $\mathcal{B}_{\alpha, \beta}(\mathcal{T}_\mathscr{P}) = \{\mathcal{S}_1, \mathcal{S}_2, \ldots, \mathcal{S}_\alpha\}$ and we pick one vertex from each section $\mathcal{S}_5$, $\mathcal{S}_{10}$, $\ldots$, $\mathcal{S}_r$. Since each section contains at least $\beta$ bags and no vertex appears in all bags of a section (except the vertices of $X$), we know that each vertex appears in at most two consecutive sections. Moreover, the number of vertices that we pick is at least $k$ since $\alpha > 5k + 8$ (assuming $k,\omega > 2$). 

Now we assume that $\mathscr{P} \neq \mathcal{T}_\mathscr{P}$. In other words, there are subtrees attached to the bags in $\mathscr{P}$. We proceed via another stage of refinement of the vertices that were already added to $X$. If there exists a vertex $x \in X$ that appears in all bags of $\mathcal{S}$, for some $\mathcal{S} \in \mathcal{B}_{\alpha, \beta}(\mathcal{T}_\mathscr{P})$, then we add $x$ to $X_1$ and we restrict $\mathscr{P}$ to the subpath of $\mathscr{P} \cap \mathcal{S}$ and recompute $\mathcal{T}_\mathscr{P}$ and $\mathcal{B}_{\alpha, \beta}(\mathcal{T}_\mathscr{P})$ with $\alpha = 6k\omega\gamma$ and $\beta = (5k^2\omega^2)$. We proceed with this refinement just as before. Hence, when no longer applicable, we know that we have $6k\omega\gamma$ sections each containing at least $(5k^2\omega^2)^{3k\omega\gamma}$ bags of $\mathscr{P}$ and we have a set $X$ of vertices that appear in all bags of $\mathscr{P}$. Moreover, we can partition $X$ into $X_1$ and $X_2$ such that vertices in $X_1$ appear in all bags of $\mathcal{T}_\mathscr{P}$ and each vertex in $X_2$ is not adjacent to at least one vertex in each section of $\mathcal{B}_{\alpha, \beta}(\mathcal{T}_\mathscr{P})$ (since otherwise that vertex would be added to $X_1$ and we would refine on the section to which it is complete). We can again either apply the multi-component reduction rule or find a component $C$ of $G - Z$ contained in $G[\mathcal{T}_\mathscr{P}]$ that contains at least $6k\omega$ sections of $\mathcal{B}_{\alpha, \beta}(\mathcal{T}_\mathscr{P})$. Let $\mathscr{P}_C$ denote the bags of $\mathscr{P}$ whose vertices belong to $C$. We update $W$ to the first bag of $\mathscr{P}_C$ and $Y$ to the last bag of $\mathscr{P}_C$. To conclude, we still need to find a $2$-independent set $R$ in $C - (N_C(W) \cup N_C(Y))$ of size at least $k$ such that each vertex in $X_2$ misses at least $k$ vertices of $R$. We proceed as before, we consider vertices in  sections $\mathcal{S}_5$, $\mathcal{S}_{10}$, $\ldots$, $\mathcal{S}_r$ but this time choosing vertices which are not adjacent to the vertices in $X_2$. Since $|X_2| \leq \omega$ and $\alpha = 6k\omega$, we can always construct such a $2$-independent set. 
\end{proof}

\begin{theorem}\label{thm:ts-chordal-bounded-clique}
\textsc{Token Sliding} on chordal graphs is fixed-parameter tractable when parameterized by $k + \omega(G)$. 
\end{theorem}

\begin{proof}
Given an instance $(G,k,I_s,I_t)$ of {\sc Token Sliding}, where $G$ is a chordal graph, we first compute a compact tree decomposition of width at most $\omega = \omega(G)$ in time $f(\omega) \cdot n^{O(1)}$~\cite{K94}. By Lemma~\ref{lem:numberbag}, we can assume that each vertex of both $I_s$ and $I_t$ belongs to at most $(3\omega+3)k^2$ bags of the tree decomposition $\mathcal{T}$. We next convert $(G,k,I_s,I_t)$ into an instance of {\sc Galactic Token Sliding} where every vertex is a planet vertex. 

We then exhaustively apply reduction rules R1, R3, R5, R6, and C1 (the adjacent black holes rule, the absorption rule, the path reduction rule, the multi-component reduction rule, and the almost black hole rule). We recompute the decomposition after each application. Once the reduction rules no longer apply we know, by Corollary~\ref{cor-bounded-bag-degree} and Corollary~\ref{cor-diam}, that each bag in the tree decomposition has at most $2^{O(\omega^2k)}$ neighbors and $G$ has diameter at most $O(k^2)$. We note that each reduction rule can be applied in polynomial time and that reduction rules R1, R3, R5, and R6 consist of either deleting bags of the decomposition or merging adjacent vertices (contracting edges) so the resulting graph remains chordal and the clique number does not increase. 

Finally, since each bag in the tree decomposition has at most $2^{O(\omega^2k)}$ neighbors and $G$ has diameter at most $O(k^2)$, we know, by Lemma~\ref{lem-chordal-reduction-application}, that the longest path in the decomposition is of length at most $\zeta$. Putting it all together, we know that after exhaustively applying all reduction rules we obtain a graph with a bounded number of vertices on which we can solve the problem via a brute-force algorithm for a total running time of $f(k, \omega) \cdot n^{O(1)}$, as needed.
\end{proof}

\section{W[1]-hardness on split graphs}\label{sec-split}

We complement our positive result for \textsc{Token Sliding} on chordal graphs of bounded clique size by showing that the problem becomes W[1]-hard (for parameter $k$) on split graphs. We give a reduction from the \textsc{Multicolored Independent Set} problem, known to be W[1]-hard~\cite{DBLP:books/sp/CyganFKLMPPS15}.

Let $(G,k)$ be an instance of \textsc{Multicolored Independent Set}. We assume that $V(G) = V_1 \cup V_2 \cup \ldots V_k$ and each $V_i = \{v^i_1, \ldots, v^i_n\}$ induces a clique. Moreover, without loss of generality, we assume that $|V_i| = n$ and therefore the total number of vertices is $kn$. We let $m_{i,j}$ and $\overline{m_{i,j}}$ denote the number of edges and non-edges between $V_i$ and $V_j$, respectively. We let $m = \sum_{i,j \in [k], i \neq j}{m_{i,j}}$ and $\overline{m} = \sum_{i,j \in [k], i \neq j}{\overline{m_{i,j}}}$. We now describe how to construct an instance $(G',k',I_s,I_t)$ of \textsc{Token Sliding}, where $G'$ is a split graph and $k' = k^2 + {k \choose 2} + 1$. 

The graph $G'$ will consist of a clique $C$ on $nk^2 + \overline{m} + 1$ vertices and an independent set $I = U \cup D$ such that $|U| = k^2 + {k \choose 2} + 1$ and $|D| = k(n(k - 1) + 1) + {k \choose 2} + 1$ ($U$ and $D$ stand for up and down and should be pictured as vertices placed above and below the clique, respectively). 
We assume that the vertices of $C$ are ordered from $c_1$ to $c_{|C|}$, the vertices of $U$ are ordered from $u_1$ to $u_{|U|}$, and the vertices of $D$ are ordered from $d_1$ to $d_{|D|}$. 
The ordering of the vertices of the different sets will be the one corresponding to their additions in the following construction.

The vertices of $C,U$ and $D$ will be created step by step each time we create a new \emph{vertex selection gadget} or a new \emph{non-edge selection gadget}. A final vertex in each set, i.e., $\{c_{|C|}, u_{|U|}, d_{|D|}\}$, will be added at the end of the construction and will constitute the \emph{switch gadget}. 


Let us now explain how $G'$ is built.
We start by creating the vertex selection gadgets which will be followed by non-edge selection gadgets, and then the last vertices will form the switch gadget (in that order). 

For every $i \le k$, the $i$-th vertex selection gadget $G_i$ consists of a new set of $k$ vertices from $U$, a new set of $nk$ vertices from $C$, and a new set of $n(k - 1) + 1$ vertices from $D$. We denote those vertices by $U_i = \{u^i_1, \ldots u^i_k\}$, $C_i = \{c^i_1, \ldots c^i_{nk}\}$, and $D_i = \{d^i_1, \ldots d^i_{n(k - 1) + 1}\}$, respectively. For $1 \leq j \leq k - 1$, we further subdivide the vertices in $\{c^i_1, \ldots c^i_{nk}\}$ and $\{d^i_1, \ldots d^i_{n(k - 1) + 1}\}$ into $k - 1$ subgroups where each subgroup is denoted by $C_{i,j} = \{c^{i,j}_1, \ldots, c^{i,j}_{n}\}$ and $D_{i,j} = \{d^{i,j}_1, \ldots, d^{i,j}_{n}\}$, respectively. Note that this leaves out a single vertex of $D_i$ which we denote by $D_{i,L} = \{d^i_{n(k - 1) + 1}\}$, $n$ vertices of $C_i$ which we denote by $C_{i,L} = \{c^{i,L}_1, \ldots, c^{i,L}_{n}\}$, and a single vertex of $U_i$ which we denote by $U_{i,L} = \{u^i_k\}$. Those $n + 2$ vertices are the \emph{lock gadget} $L_i$ of the vertex selection gadget $G_i$. 

In addition to the edges in $C$ (which indeed exists since $C$ is a clique), we add new edges for the gadgets. Let us first describe edges between vertices that belong to a same vertex selection gadget $G_i$, for every $i \le k$:
\begin{itemize}
    \item For every $1 \leq j \leq k - 1$, $u^i_j$ is connected to all vertices in $C_{i,j'}$ with $j' \geq j$ and $C_{i',j''}$ with $i'>i$.
    In other words, $u_i^j$ is connected to all vertices in $C$ except the vertices of $C$ that appear before $C_{i,j}$ (we have edges to vertices of $C$ that appear in later vertex selection or non-edge selection gadgets or the switch gadget). 
    \item For every $1 \leq q \leq n$ and $1 \leq j \leq k - 1$, $d^{i,j}_{q}$ is connected to $c^{i,j}_{q}$.
    \item For every $1 \leq q \leq n$ and $1 \leq j \leq k - 1$, $d^{i,j}_{q}$ is connected to all vertices of $C_{i,j'}$ with $j' \leq j$.
    \item For every $1 \leq k \leq k-1$ and $1 \leq q \leq n$, $j' > j$ and $q' \neq q$, we create an edge between $d^{i,j}_{q}$ and $c^{i,j'}_{q'}$.
    \item We add an edge from $U_{i,L}$ to every vertex in $C_{i,L}$ 
    \item We add an edge from $D_{i,L}$ to every vertex in $C_i$ 
    \item Finally, for every $j \le k-1$ and every $q,q' \le n$ with $q \ne q'$, we add the edge $d^{i,j}_{q}c^{i,L}_{q'}$. 
\end{itemize} 

\begin{figure}
    \centering
    \includegraphics[scale=0.8]{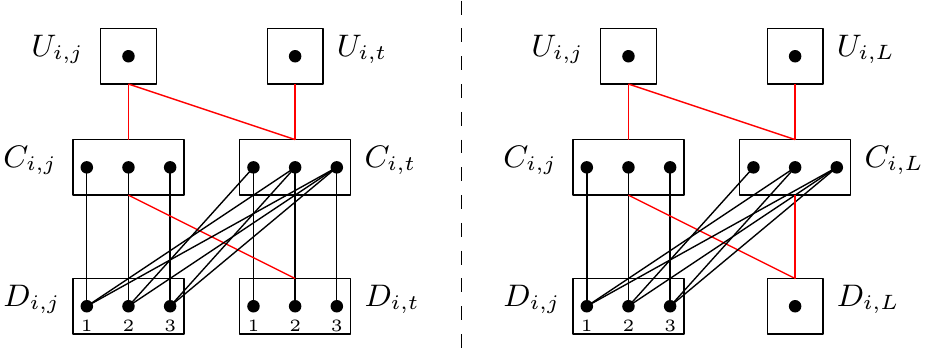}
    \caption{Edges within a selection gadget in the case $n = 3$. Connexions between sets $U_{i,j}, C_{i,j}, D_{i,j}$ and  $U_{i,t}, C_{i,t}, D_{i,t}$ with $j < t < k$ at the left, and between sets $U_{i,j}, C_{i,j}, D_{i,j}$ and $U_{i,L}, C_{i,L}, D_{i,L}$ with $j < k$ at the right. Red lines connecting two sets indicate that these sets are complete to each other. Vertices $1$ (resp. $2$, $3$) in $D_{i,j}$ and $D_{i,t}$ are representative of the same vertex of $G$.}
    \label{fig:split_selection_gadget}
\end{figure}
The connections within a vertex selection gadget are illustrated in Figure \ref{fig:split_selection_gadget}. Let us now explain how these gadgets are connected between them and two the rest of the clique vertices. For every $i \leq k$:
\begin{itemize}
    \item The vertices in $U_i$ are complete to the vertices in $C_j$ for every $j > i$ and two all vertices of $C$ that appear in later non-edge selection gadget or in the switch gadget.
    \item The vertices in $D_i$ are complete to the vertices in $D_j$ for $j < i$.
\end{itemize}
These connections are illustrated in Figure \ref{fig:split_connexion_gadgets}. We then add in $G'$ the  non-edge selection gadgets (after all vertex selection gadgets described above) as follows. Non-edge selection gadgets are added in sorted order (we sort them in lexicographic order by increasing $i$ and then by increasing $j$ with $j>i$). So we have ${k \choose 2}$ non-edge selection gadgets in total. 
A \emph{non-edge selection gadget} $G_{i,j}$ for $V_i$ and $V_j$ will consist of a new vertex of $U$ denoted by $u'_{i,j}$, a new vertex of $D$, denoted by $d'_{i,j}$, and a set of $\overline{m_{i,j}}$ new vertices in $C$, denoted by $M_{i,j}$. Each vertex in $M_{i,j}$ is denoted by $x_{uv}$ which stands for a non-edge between vertices $u$ and $v$. We add an edge from $u'_{i,j}$ to every vertex in $M_{i,j}$ and every vertex in $C$ that comes after $M_{i,j}$. We also add an edge from $d'_{i,j}$ to every vertex in $M_{i,j}$ and every vertex in $C$ that comes before $M_{i,j}$. 

We now explain how the non-edge selection gadget $G_{i,j}$ is connected to the vertex selection gadgets $G_i$ and $G_j$. Each vertex selection gadget has $k - 1$ subgroups, and each one of those subgroups is responsible for one non-edge. 
Let $D_{i,q} = \{d^{i,q}_1, \ldots, d^{i,q}_{n}\}$ denote the subgroup of $G_i$ associated with the non-edge between $i$ and $j$. Similarly, let $D_{j,p} = \{d^{j,p}_1, \ldots, d^{j,p}_{n}\}$ denote the subgroup of $G_j$ associated with the non-edge between $i$ and $j$. Recall that here $q$ and $p$ are between $1$ and $n$, i.e., we assume that each vertex in $D_{i,q}$ or $D_{j,p}$ is a copy of a vertex from $u \in V_i$ or $v \in V_j$, respectively. For each $x_{uv} \in M_{i,j}$ we add an edge from $x_{uv}$ to every vertex in $D_{i,q}$ and $D_{j,q}$ except the vertices associated with $u$ and $v$, respectively. 

We conclude the construction by adding an edge from $u_{|U|}$ to $c_{|C|}$ and an edge from $d_{|D|}$ to $c_{|C|}$, which constitutes the switch gadget. 
We let $I_s = U$ and $I_t = (U \setminus \{u_{|U|}\}) \cup \{d_{|D|}\}$.

\begin{figure}
    \centering
    \includegraphics[scale=0.8]{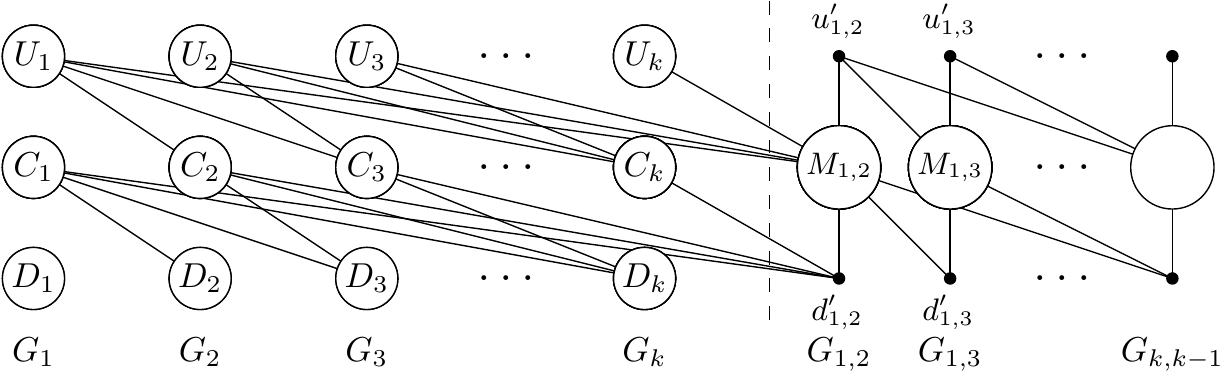}
    \caption{Representation of the vertex selection gadgets and the non-edge selection gadgets. Black line between sets (including singletons) indicates that the sets are complete to each other. Some edges have been omitted for the sake of clarity: only the connections between vertex selection gadgets and the first non-edge selection gadget are represented. Furthermore, the connections between the sets $M$ and $D$ representing the non-edges have been omitted. 
    }
    \label{fig:split_connexion_gadgets}
\end{figure}

\begin{lemma}\label{lem-split-fwd}
If $(G,k)$ is a yes-instance of \textsc{Multicolored Independent Set} then $(G',k',I_s,I_t)$ is a yes-instance of \textsc{Token Sliding}.
\end{lemma}

\begin{proof}
Let $\{v^1_{q_1}, v^2_{q_2}, \ldots, v^k_{q_k}\}$ be a multicolored independent set of $G$, where $v^1_{q_1} \in V_1, \ldots, v^k_{q_k} \in V_k$. We first show how to slide every token in $U$ (except the last) to a vertex in $D$. Following that we can slide the token on $u_{|U|}$ to $d_{|D|}$ and then finally bring back the other tokens back up to $U$ (where we reverse the order in which they were slid from $U$ to $D$).

Recall that a vertex selection gadget $G_i$ for $V_i$ will consist of the $i$th set of $k$ vertices from $U$, the $i$th set of $nk$ vertices from $C$, and the $i$th set of $n(k - 1) + 1$ vertices from $D$. We denote those vertices by $U_i = \{u^i_1, \ldots u^i_k\}$, $C_i = \{c^i_1, \ldots c^i_{nk}\}$, and $D_i = \{d^i_1, \ldots d^i_{n(k - 1) + 1}\}$, respectively. For $1 \leq j \leq k - 1$, we further subdivide the vertices in $\{c^i_1, \ldots c^i_{nk}\}$ and $\{d^i_1, \ldots d^i_{n(k - 1) + 1}\}$ to $k - 1$ subgroups where each subgroup is denoted by $C_{i,j} = \{c^{i,j}_1, \ldots, c^{i,j}_{n}\}$ and $D_{i,j} = \{d^{i,j}_1, \ldots, d^{i,j}_{n}\}$, respectively. Starting from $G_1$ up to $G_k$, we proceed as follows. In $G_i$, we slide from $u^i_1$ to $c^{i,1}_{q_i}$, where $c^{i,1}_{q_i}$ corresponds to vertex $v^i_{q_i}$ of the multicolored independent set. We then slide from $c^{i,1}_{q_i}$ to $d^{i,1}_{q_i}$, such an edge exists by construction. We repeat for each $j$, that is, we then slide from $u^i_2$ to $c^{i,2}_{q_i}$ to $d^{i,2}_{q_i}$, all the way to sliding from $u^i_{k-1}$ to $c^{i,k-1}_{q_i}$ to $d^{i,k-1}_{q_i}$. To see why those slides are valid, note that every time we slide a token  to the clique from $U$ then all previous vertices of $U$ have already slid down to $D$ (vertices that appear earlier in $U$ are already in $D$). Moreover, within each vertex selection gadget, we always slide the $k - 1$ tokens to the same representative vertex of the multicolored independent set vertex, which implies that those vertices of $C$ and $D$ are pairwise non-adjacent. Formally, the graph induced by $c^{i,1}_{q_i}, \ldots, c^{i,k-1}_{q_i}$ and $d^{i,1}_{q_i}, \ldots, d^{i,k-1}_{q_i}$ is a matching between the pairs $c^{i,j}_{q_i}$ and $d^{i,j}_{q_i}$. We complete the transformation of $G_i$ by sliding from $u^i_k$ to 
$c^{i,L}_{q_i}$ and then from $c^{i,L}_{q_i}$ to $D_{i,L}$. Again, by construction, there are no edges between $c^{i,L}_{q_i}$ and any of the vertices $d^{i,1}_{q_i}, \ldots, d^{i,k-1}_{q_i}$. 

We now show how to slide the tokens in $U$ from non-edge selection gadgets down to $D$. As before, we proceed in order (following the order in which the gadgets were created). Recall that a non-edge selection gadget $G_{i,j}$ for $V_i$ and $V_j$ will consist of a vertex of $U$ which we denote by $u'_{i,j}$, a vertex of $D$ which we denote by $d'_{i,j}$, and a group of $\overline{m_{i,j}}$ vertices from $C$ which we denote by $M_{i,j}$. 
Each non-edge selection gadget $G_{i,j}$ interacts (has edges to) one subgroup in $G_i$ and one subgroup in $G_j$. 
Let $D_{i,q} = \{d^{i,q}_1, \ldots, d^{i,q}_{n}\}$ denote the subgroup of $G_i$ associated with the non-edge between $i$ and $j$. Similarly, let $D_{j,p} = \{d^{j,p}_1, \ldots, d^{j,p}_{n}\}$ denote the subgroup of $G_j$ associated with the non-edge between $i$ and $j$. Let $d^{i,q}_{\star}$ denote the vertex of $D_{i,q}$ containing a token and let $d^{j,p}_{\star}$ denote the vertex of $D_{j,p}$ containing a token. Since both these vertices correspond to vertices of the multicolored independent set we know that there is a non-edge between them in the original graph. Hence, by construction, there exists a vertex in $M_{i,j}$ which is not adjacent to neither $d^{i,q}_{\star}$ nor $d^{j,p}_{\star}$. We can therefore slide the token from $u'_{i,j}$ to that vertex in $C$ then down to $d'_{i,j}$. We repeat the same procedure for every non-edge selection gadget. 

Finally, after having moved all tokens in vertex selection gadgets and edge selection gadgets down to $D$, we slide the token on $u_{|U|}$ to $c_{|C|}$ and then down to $d_{|D|}$. After that all other tokens in $D$ can be slid back to $U$ in reverse order, which completes the proof. 
\end{proof}

\begin{lemma}\label{lem-hardness-order}
Let $U = \{u_1, \ldots, u_{|U|}\}$. If $(G',k',I_s,I_t)$ is a yes-instance then, in a reconfiguration sequence from $I_s$ to $I_t$, before the token on $u_i$ can move all tokens on $u_j$, $j < i$, must have moved to $D$. 
\end{lemma}

\begin{proof}
Consider $u_{|U|}$ whose only neighbor in $C$ is $c_{|C|}$ whose only neighbor in $D$ is $d_{|D|}$. Every other vertex in $U$ is connected to $c_{|C|}$ and we can only have one token in the clique at any point in time. Hence, before $u_{|U|}$ can slide down to $C$ all other tokens must already be in $D$. 

Similarly, consider a vertex $u_i \in U$. If $u_i$ belongs to either a vertex selection gadget or a non-edge selection gadget then every vertex of $U$ appearing before $u_i$ is connected to all neighbors of $u_i$ in $C$. Hence, before $u_i$ can slide down to the clique all other tokens on $u_j$, $j < i$, must have moved to $D$. 
\end{proof}

Given a yes-instance $(G',k',I_s,I_t)$ and a reconfiguration sequence $\mathcal{R} = \langle I_0,I_1,\ldots,I_{\ell-1},I_\ell \rangle$ from $I_s = I_0$ to $I_t = I_\ell$, we say that $\mathcal{R}$ is \emph{non-redundant} whenever $\mathcal{R}$ does not contain two identical sets at different indices, i.e., $I_p = I_q$ and $q \neq p$. The sequence is \emph{redundant} otherwise. Note that given any reconfiguration sequence we can easily transform it into a non-redundant one (by simply truncating appropriately).

\begin{corollary}\label{cor-hardness-canonical}
Let $(G',k',I_s,I_t)$ be a yes-instance. Then in a redundant or non-redundant reconfiguration sequence $\mathcal{R} = \langle I_0,I_1,\ldots,I_{\ell-1},I_\ell \rangle$ from $I_s = I_0$ to $I_t = I_\ell$ there exists a first set $I_\alpha$, $\alpha > 0$, and a first set $I_\beta$, $\beta > \alpha$, such that the following holds:
\begin{itemize}
    \item $\{u_{1}, \ldots, u_{k^2}\} \cap I_\alpha = \emptyset$, $\{u_{k^2 + 1}, \ldots, u_{|U|}\} \subset I_\alpha$, and $I_\alpha \setminus \{u_{k^2 + 1}, \ldots, u_{|U|}\} \subseteq D$;  
    \item $\{u_{1}, \ldots, u_{|U| - 1}\} \cap I_\beta = \emptyset$, $\{u_{|U|}\} \subset I_\beta$, and $I_\beta \setminus u_{|U|} \subseteq D$.
\end{itemize}
\end{corollary}

\begin{proof}
Let $I_\alpha$ denote the independent set obtained after all tokens initially in vertex selection gadgets are now in $D$ and before any token from non-edge selection gadgets moves. Note that by Lemma \ref{lem-hardness-order} there indeed exists such a set in the sequence, right before the first move of the token on the first non-edge selection gadget $u_{k^2+1}$.
Let then $I_\beta$ denote the independent set occurring before the token on $u_{|U|}$ slides to $c_{|C|}$ for the first time. 
Both sets satisfy the required properties by Lemma~\ref{lem-hardness-order}.
\end{proof}

Given an instance $(G',k',I_s,I_t)$, we say that a vertex selection gadget $G_i$ is \emph{well-behaved} whenever the existence of a non-redundant reconfiguration sequence from $I_s$ to $I_t$ implies that the slides ordered below are the only slides of the tokens in $G_i$ that occur between $I_s$ and $I_\alpha$ (inclusive) and that these tokens remain in places until $I_\beta$ (inclusive). Let $q$ be some integer between $1$ and $n$.
\begin{itemize}
    \item The token on $u^{i}_1$ slides to some vertex in $C_{i,1}$ (potentially moving around in $C_{i,1}$) then slides to $c^{i,1}_q$ and then finally to $d^{i,1}_q$ (or the token goes directly from $u^{i}_1$ to $c^{i,1}_q$ to $d^{i,1}_q$);
    \item For all $2 \leq j \leq k-1$, the token on $u^{i}_j$ slides to $c^{i,j}_q$ and then to $d^{i,j}_q$;
    \item The token on $u^{i}_{k}$ slides to $c^{i,L}_q$ and then to $D_{i,L}$.
\end{itemize}

Note that the token moving from $u^i_1$ is the only one that can do several moves in $C_i$ as there are initially no token on $D_i$. Once this token moves to a vertex $d^{i,1}_q$ in $D_i$, the other vertices on $u^i_j$ for $j \geq 2$ have no choice but to move on $c^{i,j}_q$ by construction of the vertex-selection gadgets.

\begin{lemma}\label{lem-hardness-well-behaved}
If $(G',k',I_s,I_t)$ is a yes-instance then in a non-redundant reconfiguration sequence from $I_s$ to $I_t$ every vertex selection gadget is well-behaved.
\end{lemma}

\begin{proof}
If $(G',k',I_s,I_t)$ is a yes-instance then we know that we can find a non-redundant sequence from $I_s$ to $I_t$ and, by Corollary~\ref{cor-hardness-canonical}, we know that there exists $I_\alpha$, $\alpha > 0$, and $I_\beta$, $\beta > \alpha$, such that $\{u_{1}, \ldots, u_{k^2}\} \cap I_\alpha = \emptyset$, $\{u_{k^2 + 1}, \ldots, u_{|U|}\} \subset I_\alpha$, $I_\alpha \setminus \{u_{k^2 + 1}, \ldots, u_{|U|}\} \subseteq D$, $\{u_{1}, \ldots, u_{|U| - 1}\} \cap I_\beta = \emptyset$, $\{u_{|U|}\} \subset I_\beta$, and $I_\beta \setminus u_{|U|} \subseteq D$. In other words, $I_\alpha$ is the set where all tokens initially on the vertex-selection gadgets are now in $D$ and where all other tokens are still in there initial position. As for $I_\beta$, it is the set where all tokens have moved to $D$, except for the one initially on the switch gadget which is still on its initial position.

Assume that there exists a vertex selection gadget which is not well-behaved and let $G_i$ be the first such gadget, i.e, all gadgets $G_j$, $j < i$, are well-behaved. 
Consider the independent set $I'$ where the token on $u^i_{k}$ (the token on the lock gadget of $G_i$) slides to some vertex in $C$ for the first time (which exists by Lemma \ref{lem-hardness-order}). Since every gadget before $G_i$ is well-behaved we know that none of the tokens of $G_i$ could have slid out of $G_i$ to an earlier gadget (the lock gadget prevents that with a token on $D_{i,L}$ which is complete to $C_{i'}$ for every $i' < i$). Moreover, since $u^i_{k}$ is sliding into $D$ for the first time, we know that all the later tokens of $U$ are still in place and will prevent tokens of $G_i$ to slide to later gadgets. Putting it all together, we know that in $I'$ all tokens of $G_i$ are still contained in $G_i$. Now consider the independent set $I''$ where the token on $u^i_{k}$ (which is now in $C$) slides to some vertex in $D$ for the first time. We claim that when this slide happens we must have exactly one vertex in each group $D_{i,j} = \{d^{i,j}_1, \ldots, d^{i,j}_{n}\}$, $1 \leq j \leq k - 1$. This follows by construction since having tokens on $d^{i,j}_{q}$ and $d^{i,j}_{q'}$, $q \neq q'$, implies that those tokens are adjacent to all vertices of $C_{i,L} = \{c^{i,L}_1, \ldots, c^{i,L}_{n}\}$ (blocking the token on $u^i_{k}$). Also by construction, we know that the $k - 1$ tokens must be on vertices of the form $d^{i,1}_q$, $\ldots$, $d^{i,k - 1}_q$, for some $q$, as otherwise two of those vertices will be adjacent to all vertices of $C_{i,L}$.
Hence, after $I''$, the only possible slide for the token originally on $u^i_{k}$ (we have a non-redundant sequence) is to slide to $D_{i,L}$. Let $I'''$ be the independent set where the first token of $G_{i+1}$ slides to $C$. We have shown that up until $I'''$ the vertex selection gadget $G_i$ must be well-behaved. Now assume that some token in $G_i$ moves after $I'''$ (and before $I_\beta$). For that to happen, it must be the case that the token on the lock gadget slides back to $C$ (and then to some other vertex). But for this to happen it must be the case that all tokens that slid to some token of $D$ later (after $I'''$) must slide back to $U$ (since every vertex in $D$ is connected to all earlier vertices in $C$). Since furthermore the tokens on $D_{i,j}$ for $j \leq k-1$ and the tokens on $D_i'$ for $i' < i$ cannot move before the token on $D_{i,L}$ moves again to $C_i$. It follows that we obtain the independent set $I'$ again and that the sequence is redundant.
\end{proof}

\begin{lemma}\label{lem-split-back}
If $(G',k',I_s,I_t)$ is a yes-instance of \textsc{Token Sliding} then $(G,k)$ is a yes-instance of \textsc{Multicolored Independent Set}. 
\end{lemma}

\begin{proof}
If $(G',k',I_s,I_t)$ is a yes-instance then we know that we can find a non-redundant reconfiguration sequence $\mathcal{R}$ from $I_s$ to $I_t$ and, by Corollary~\ref{cor-hardness-canonical}, we know that there exists $I_\alpha$, $\alpha > 0$, and $I_\beta$, $\beta > \alpha$, such that $\{u_{1}, \ldots, u_{k^2}\} \cap I_\alpha = \emptyset$, $\{u_{k^2 + 1}, \ldots, u_{|U|}\} \subset I_\alpha$, $I_\alpha \setminus \{u_{k^2 + 1}, \ldots, u_{|U|}\} \subseteq D$, $\{u_{1}, \ldots, u_{|U| - 1}\} \cap I_\beta = \emptyset$, $\{u_{|U|}\} \subset I_\beta$, and $I_\beta \setminus u_{|U|} \subseteq D$. 
By Lemma~\ref{lem-hardness-well-behaved}, we know that every vertex selection gadget is well-behaved in this sequence $\mathcal{R}$. 

To complete the proof, it is enough to observe that in $I_\alpha$ the tokens in each vertex selection gadget are actually placed on the same vertex of $D_{i,j} = \{d^{i,j}_1, \ldots, d^{i,j}_{n}\}$. In other words, there exists $q$ such that each token in $D_{i,j}$ (for each $0 \leq j \leq k - 1$) is on $d^{i,j}_{q}$. To see why this is enough, note that each $d^{i,j}_{q}$ corresponds to the same vertex in $v_q \in V_i$.
So let $v_q \in V_i$ and $v_p \in V_{i'}$ denote two such vertices. Since all tokens on non-edge selection gadgets must slide after vertex selection gadgets, then each such non-edge selection gadgets guarantees that there is no edge in $G$ between $v_q$ and $v_p$, as needed. 
\end{proof}

Combining Lemma~\ref{lem-split-fwd} and~\ref{lem-split-back}, we get the following theorem. 

\begin{theorem}\label{thm-split-hardness}
\textsc{Token Sliding} parameterized by $k$ is W[1]-hard on split graphs. 
\end{theorem}

We note that Theorem~\ref{thm-split-hardness} can be modified to show that \textsc{Token Sliding} parameterized by $k + \ell$ is W[1]-hard on split graphs, where $\ell$ is the length of a reconfiguration sequence. 


\bibliographystyle{plainurl}
\bibliography{biblio}

\end{document}